\documentclass[a4paper,11pt]{article}

\usepackage{a4wide}
\usepackage{graphicx}
\usepackage{anysize}
\usepackage{wrapfig}
\usepackage{amssymb}
\usepackage{url}
\usepackage{comment}
\usepackage{array}
\usepackage{multirow}
\usepackage{rotating}
\usepackage{mathptmx}

\newcommand{\ie}{{\em i.e.}}
\newcommand{\ens}[1]{\ensuremath{\{#1\}}}
\newcolumntype{M}[1]{>{\centering}m{#1}}
\newcolumntype{C}{>{$\displaystyle}c<{$}}
\newtheorem{definition}{Definition}
\newtheorem{lemma}{Lemma}
\newcommand{\qed}{\hfill$\Box$ \vspace{0.5 cm}}
\newenvironment{proof}{{\it Proof: }}{\qed}

\begin{document}

\title{Internal links and pairs as a new tool\\ for the analysis of bipartite complex networks}

\author{
\begin{minipage}{14cm}
\center
Oussama Allali, Lionel Tabourier, Cl{\'e}mence Magnien and Matthieu Latapy\\
LIP6, Universit\'e Pierre et Marie Curie, 4 Place Jussieu, F-75252 Paris, France\\ 
\url{firstname.lastname@lip6.fr}
\end{minipage}
}


\date{}

\maketitle

\begin{abstract}
Many real-world complex networks are best modeled as bipartite (or 2-mode) graphs, where nodes are divided into two sets with links connecting one side to the other.
However, there is currently a lack of methods to analyze properly such graphs as most existing measures and methods are suited to classical graphs.
A usual but limited approach consists in deriving 1-mode graphs (called projections) from the underlying bipartite structure, though it causes important loss of information and data storage issues. 
We introduce here {\em internal links} and {\em pairs} as a new notion useful for such analysis: it gives insights on the information lost by projecting the bipartite graph.
We illustrate the relevance of theses concepts on several real-world instances illustrating how it enables to discriminate behaviors among various cases when we compare them to a benchmark of random networks. 
Then, we show that we can draw benefit from this concept for both modeling complex networks and storing them in a compact format.
\end{abstract}


\section{Introduction}

Many real-world networks have a natural bipartite (or 2-mode)
structure and so are best modeled by bipartite graphs: two kinds of
nodes coexist and links are between nodes of different kinds only.
Typical examples include biological networks in which proteins are
involved in biochemical reactions, occurrence of words in
sentences of a book, authoring of scientific papers, file-provider
graphs where each file is connected to the individuals providing it,
and many social networks where people are members of groups like
directory boards. See \cite{newman2001random,latapy2008basic} for more
examples.

The classical approach for studying such graphs is to turn them into classical (non-bipartite) graphs using the notion of {\em projection}: considering only one of the two types of nodes and linking any two nodes if they share a neighbor  
in the bipartite graph. This leads for instance to cooccurrence graphs, where two words are linked if they appear in a same sentence, coauthoring graphs, where two researchers are linked if they are authors of a same paper, interest graphs where individuals are linked together if they provide a same file, etc.

This approach however has severe drawbacks \cite{latapy2008basic}. In particular, it leads to huge projected graphs, and much information is lost in the projection. There is therefore much interest in methods that would make it possible to study bipartite graphs directly, without resorting to projection. Despite previous efforts to develop such methods \cite{lind2005cycles,latapy2008basic,Zweig2011Systematic}, much remains to be done in this direction.

We propose in this paper a new notion, namely {\em internal links} and {\em pairs}, useful for the analysis of real-world bipartite graphs. We introduce it in Section~\ref{sec:def}, then present some datasets in Section~\ref{sec:data} that we use as typical real-world cases which we analyze in Section~\ref{sec:meas} with regard to our new notion. We explore a more algorithmic perspective in Section~\ref{sec:delet}.

\section{Internal pairs and links \label{sec:def}}

Let us consider a bipartite graph $G = (\bot,\top,E)$, with $E \subseteq \bot \times \top$. We call nodes in $\bot$ (resp. $\top$) the bottom (resp. top) nodes. We denote by $N(u)=\{v \in (\bot \cup \top),\ (u,v) \in E\}$ the neighborhood of any node $u$. We extend this notation to any set $S$ of nodes as follows: $N(S) = \cup_{v\in S} N(v)$.

The $\bot$-projection of $G$ is the graph $G_\bot =(\bot,E_\bot)$ in which $(u,v) \in E_\bot$ if $u$ and $v$ have at least one neighbor in common (in $G$): $N(u) \cap N(v) \neq \emptyset$.
We will denote by $N_\bot(u)$ the neighborhood of a node $u$ in $G_{\bot}$: $N_\bot(u)=\{v \in \bot, \ (u,v) \in E_\bot\}=N(N(u))$. The $\top$-projection $G_\top$ is defined dually.

For any pair of nodes $(u,v) \notin E$, we denote by $G + (u,v)$ the graph $G' = (\bot,\top,E \cup \{(u,v)\})$ obtained by adding the new link $(u,v)$ to $G$. For any link $(u,v) \in E$, we denote by $G - (u,v)$ the graph $G' = (\bot,\top,E \setminus \{(u,v)\})$ obtained by removing link $(u,v)$ from $G$.

\begin{definition}[internal pairs]
A pair of nodes $(u,v)$ with $(u,v) \notin E$ is a {\em $\bot$-internal} pair of $G$ if the $\bot$-projection of $G' = G + (u,v)$ is identical to the one of $G$.
We define $\top$-internal pairs dually.
\end{definition}

\begin{figure}[h!]
\centering
\begin{tabular}{M{4.5cm}M{4.5cm}M{4.5cm}M{4.5cm}}
\includegraphics[width=0.25\textwidth]{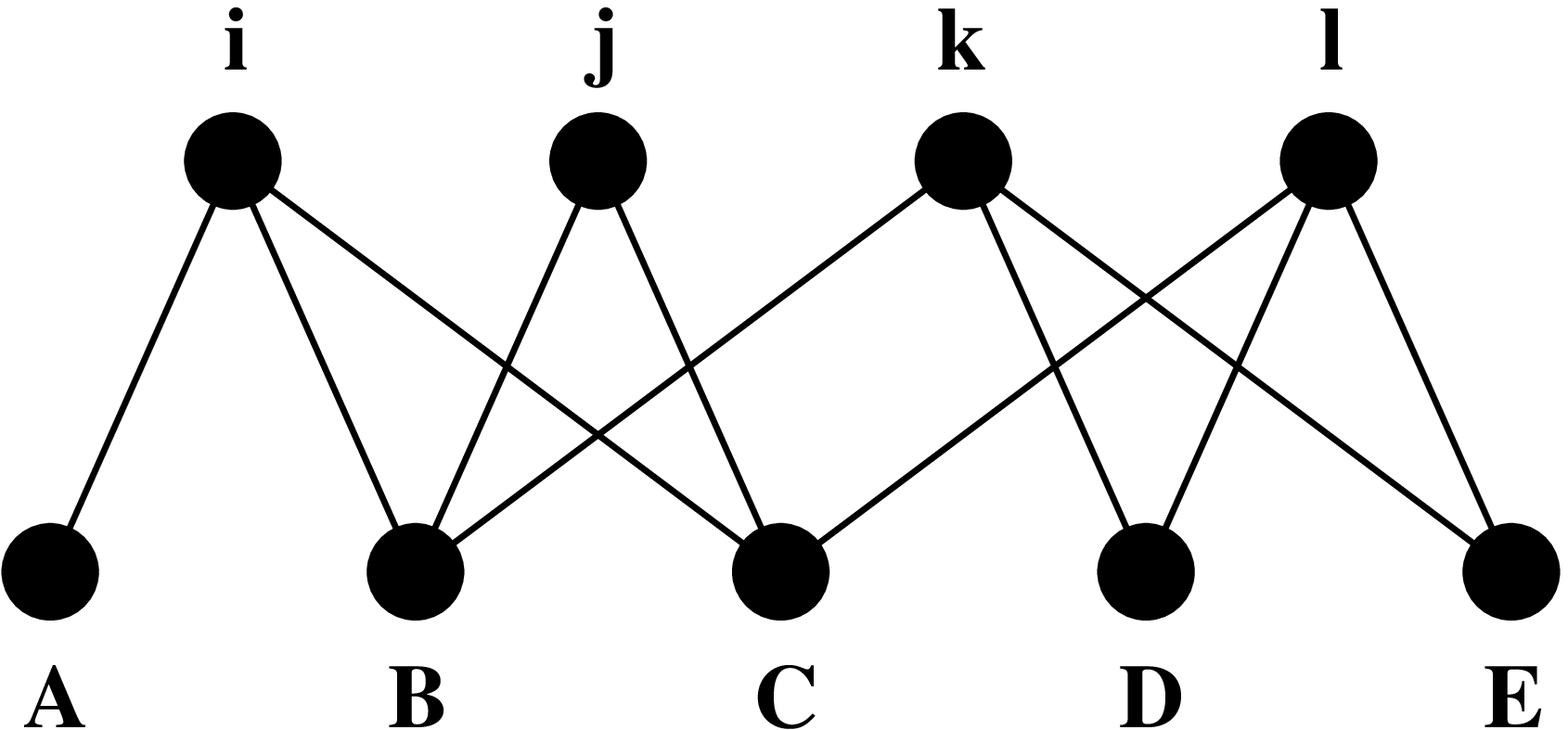} & 
\includegraphics[width=0.25\textwidth]{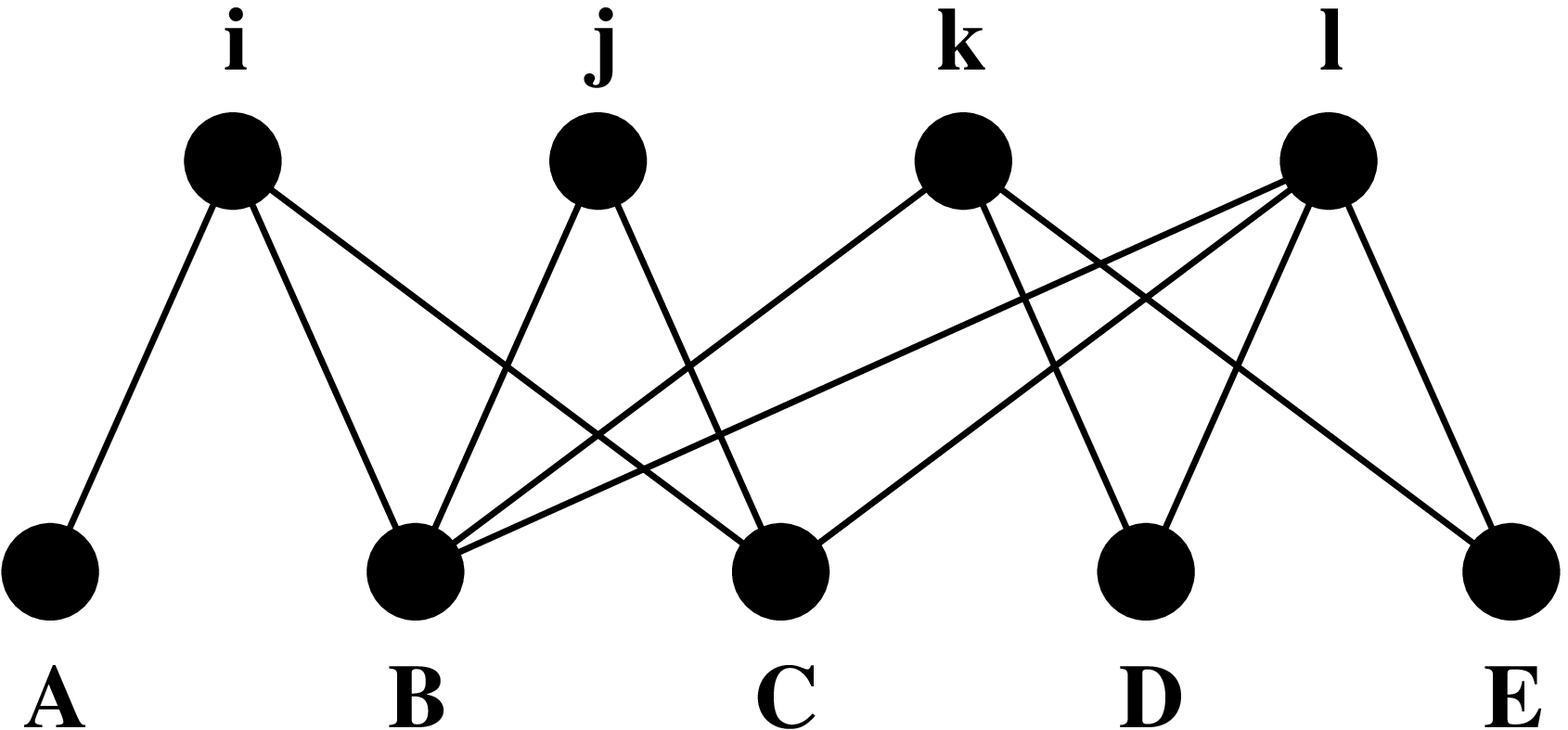} & 
\includegraphics[width=0.17\textwidth]{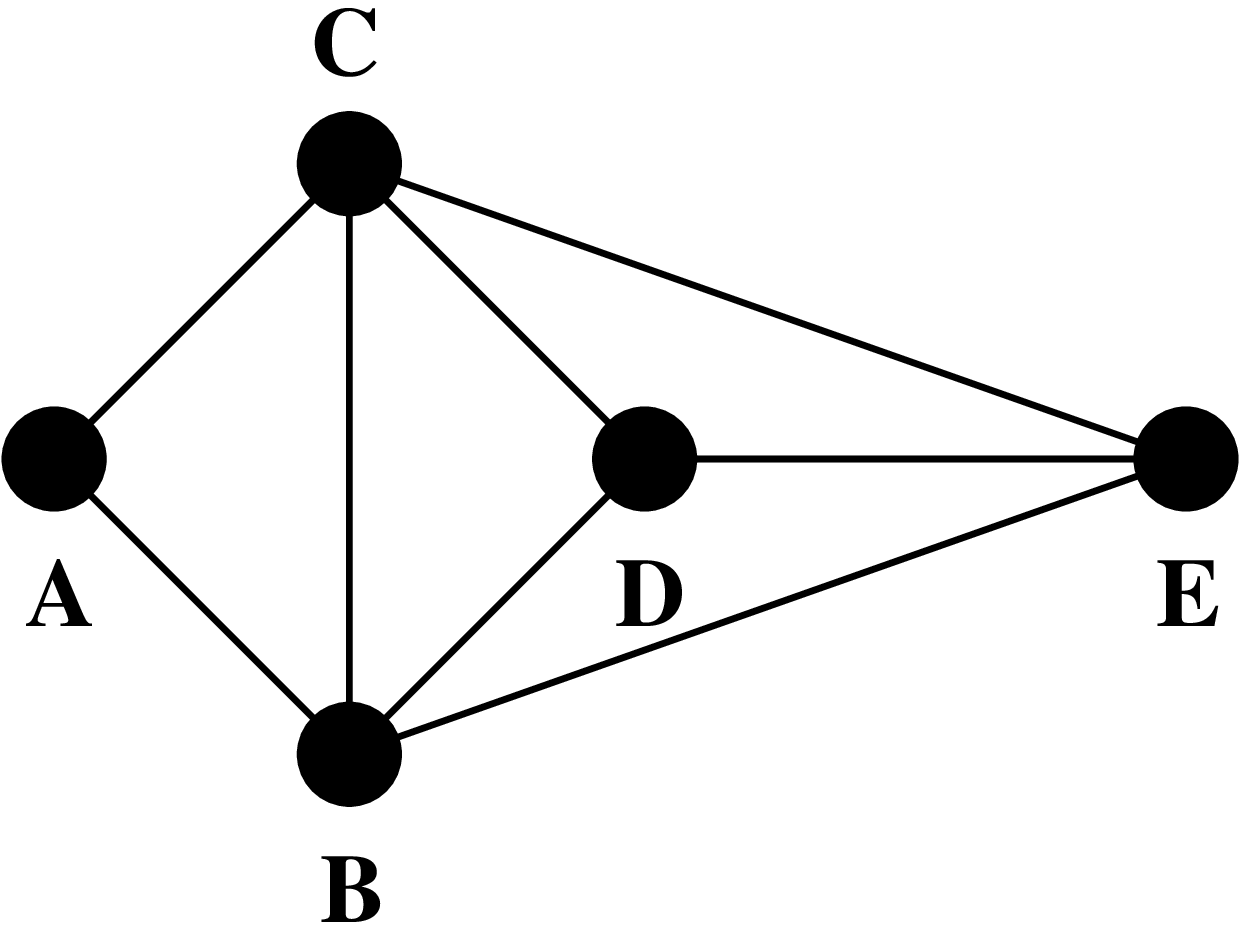}
\tabularnewline
$G$ & $G'=G+(B,l)$ & $G'_\bot = G_\bot$ 
\tabularnewline
\end{tabular}
\caption{\textbf{Example of $\bot$-internal pair}. 
Left to right: a bipartite graph $G$, the bipartite graph $G'$ obtained by adding link $(B,l)$ to $G$, and the $\bot$-projection of these two graphs. As $G'_{\bot} = G_{\bot}$, $(B,l)$ is a $\bot$-internal pair of $G$.}
\label{Ex-internal-pair}
\end{figure}

\begin{definition}[internal links]
A link $(u,v) \in E$ is a {\em $\bot$-internal} link of $G$ if the $\bot$-projection of $G' = G - (u,v)$ is identical to the one of $G$.
We define $\top$-internal links dually.
\end{definition}

\begin{figure}[h!]
\centering
\begin{tabular}{M{4.5cm}M{4.5cm}M{4.5cm}M{4.5cm}}
\includegraphics[width=0.25\textwidth]{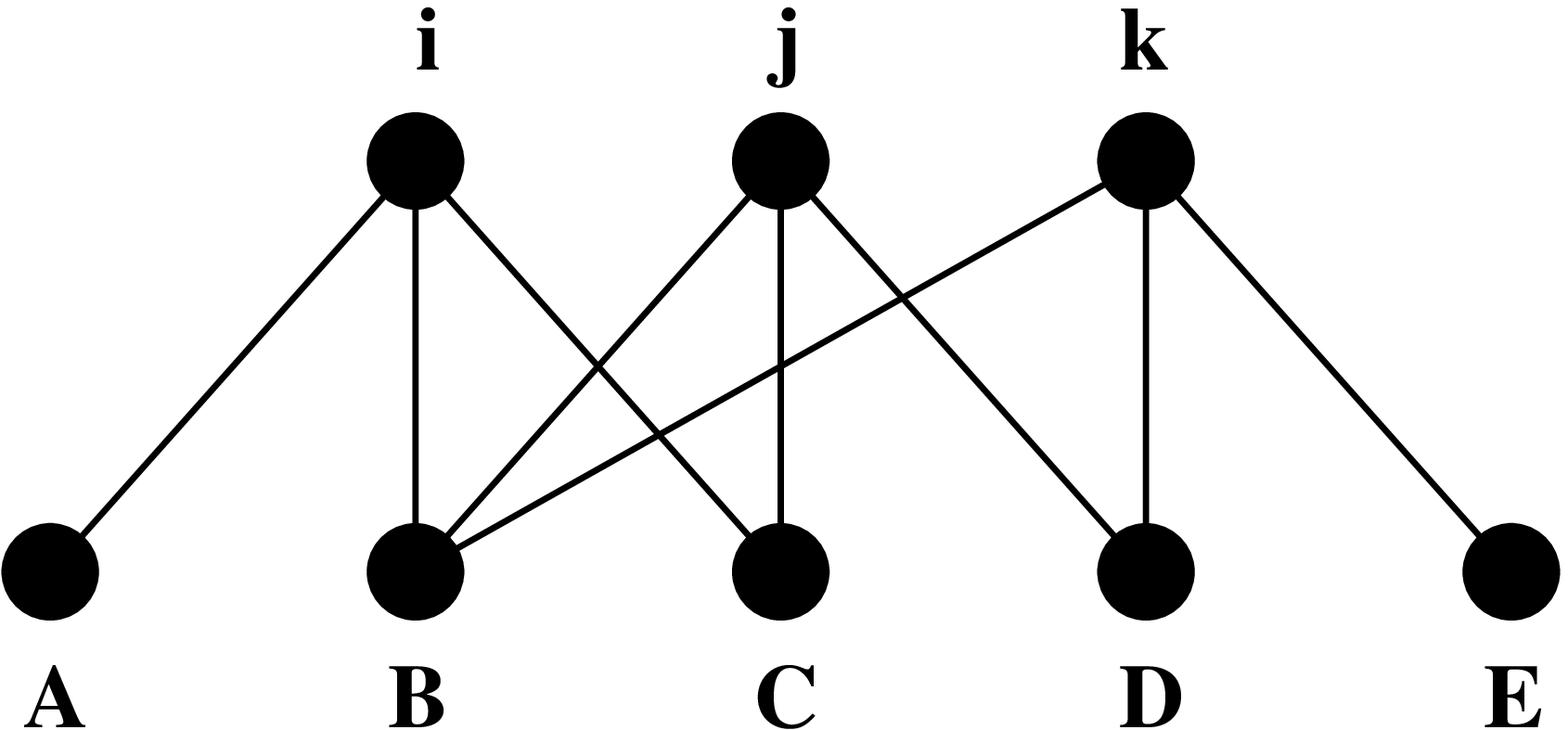} & 
\includegraphics[width=0.25\textwidth]{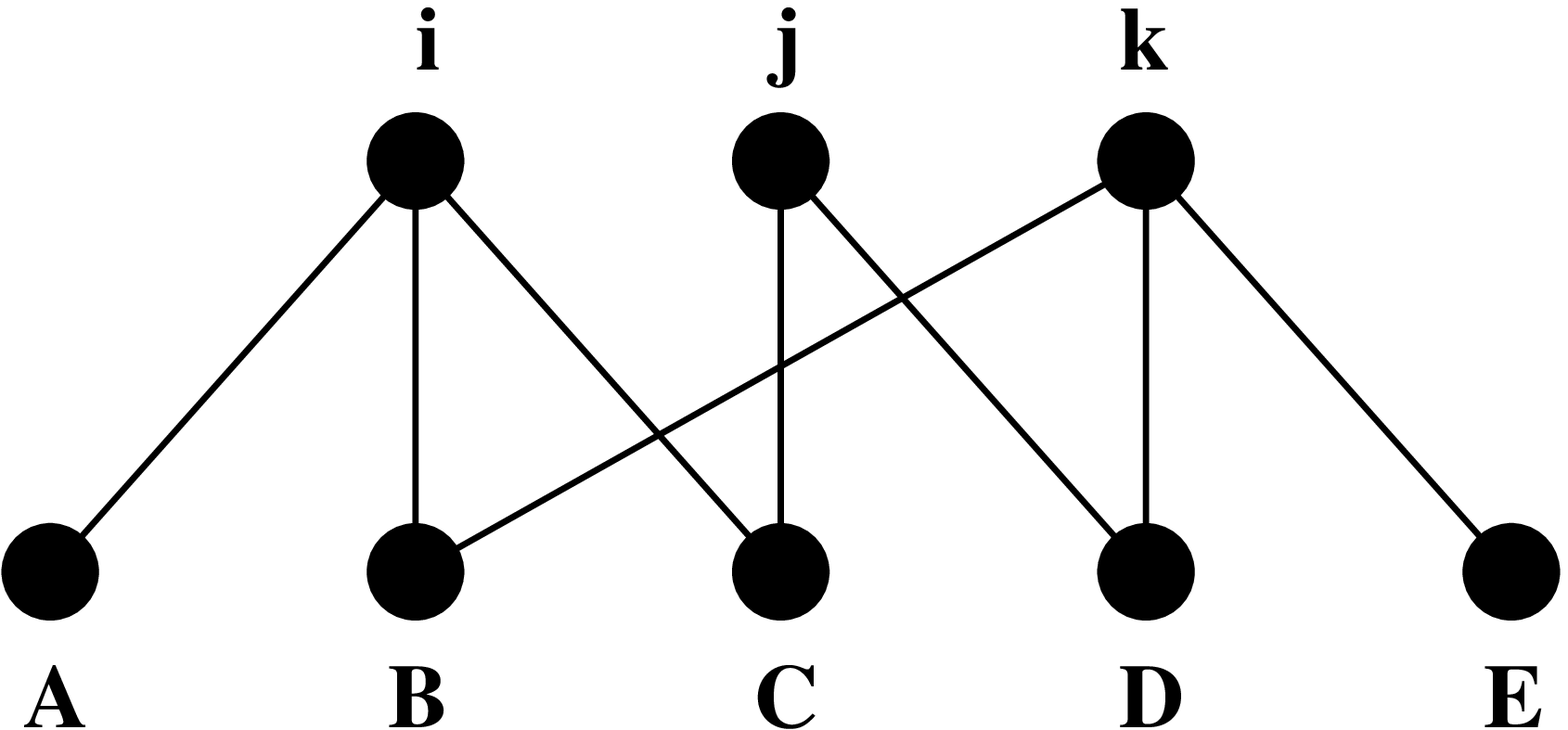} & 
\includegraphics[width=0.17\textwidth]{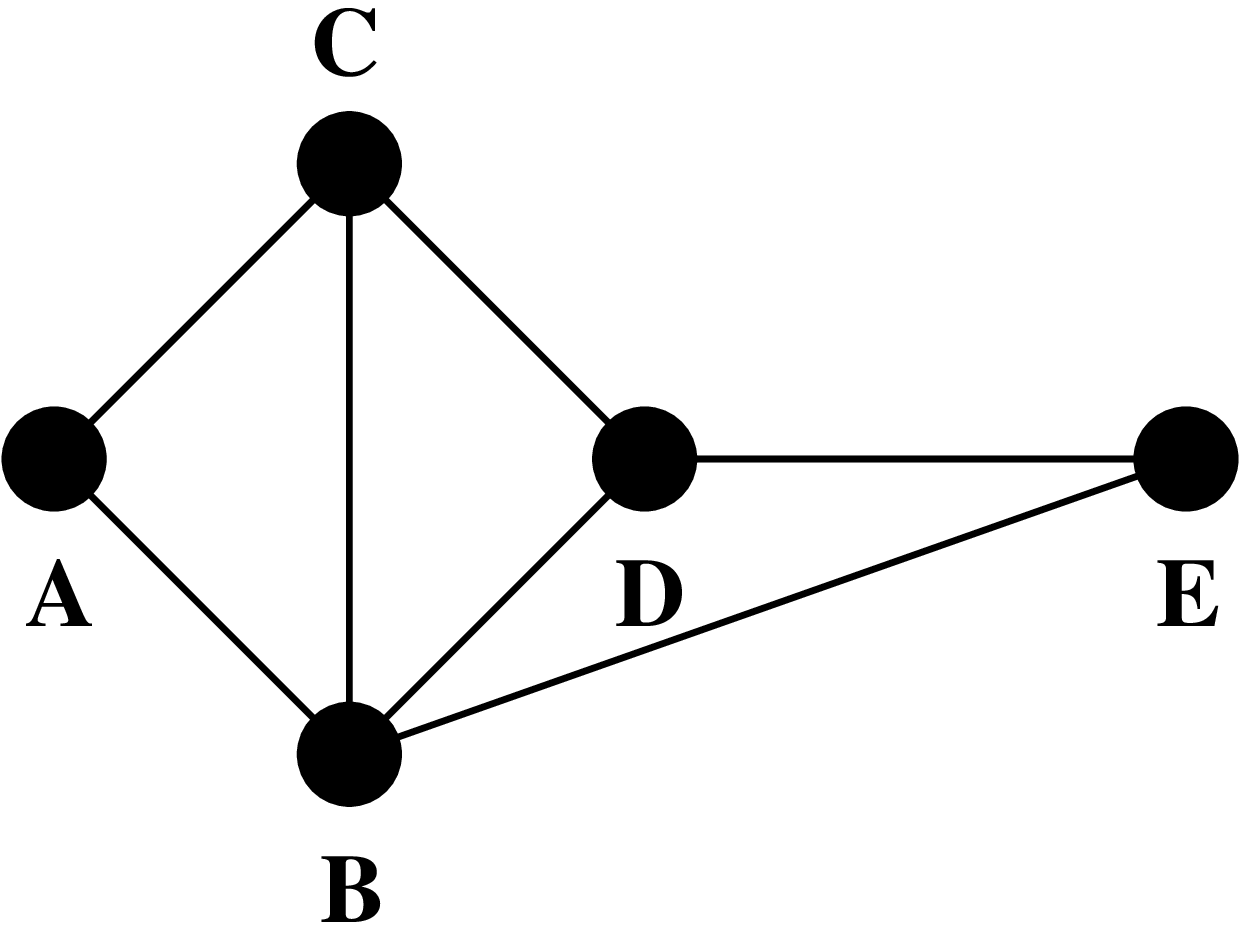}
\tabularnewline
$G$ & $G'=G-(B,j)$ & $G'_\bot = G_\bot$ 
\tabularnewline
\end{tabular}
\caption{\textbf{Example of $\bot$-internal link}. 
Left to right: a bipartite graph $G$, the bipartite graph $G'$ obtained by removing link $(B,j)$ from $G$, and the $\bot$-projection of these two graphs. As $G'_{\bot} = G_{\bot}$, $(B,j)$ is a $\bot$-internal link of $G$.}
\label{Ex-internal-link}
\end{figure}

In other words, $(u,v)$ is a $\bot$-internal pair of $G$ if adding the new link $(u,v)$ to $G$ does not change its $\bot$-projection; it is a $\bot$-internal link if removing link $(u,v)$ from $G$ does not change its $\bot$-projection. See Figure~\ref{Ex-internal-pair} and \ref{Ex-internal-link} for examples.

The notion of internal link is related to the redundancy of a node~\cite{latapy2008basic},
defined  for any node $v$ as the fraction of
pairs in $N(v)$ that are still linked together in the projection of the
graph $G'$ obtained from $G$ by removing $v$ and all its links
(all these pairs are linked in $G_\bot{}$). 
There
is however no direct equivalence between the two notions. The
redundancy is a node-oriented property: it gives a value for each
node, while the notion of internal links and pairs is link-oriented.
As illustrated on
Figure~\ref{fig:redundancy}, nodes exhibiting the same fraction of
internal links may have different redundancies, and conversely two
nodes having the same redundancy may correspond to different internal
connectivity patterns. 
It is possible to classify the links of each node as $\bot$ and $\top$-internal or not;
this induces a notion of {\em $\bot$-internal degree} of a node (resp. {\em $\top$-internal degree}), 
which is its number of internal links (see next
section).

\begin{figure}[h!]
\centering
\includegraphics[width=0.4\textwidth]{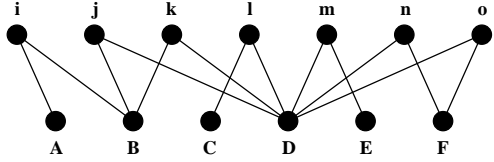} 
\caption{\label{fig:redundancy}\textbf{Redundancy versus internal links.}
In this graph, B and D have the same fraction of $\bot$-internal links ($ \frac{2}{3}$) while having different redundancies (resp. $ \frac{1}{3}$ and $ \frac{2}{15}$).
}
\end{figure}

We now give a characterization of internal links, which does not explicitly rely on the projection anymore and provides another point of view on this notion.

\begin{lemma}
\label{lemme-in}
A link $(u,v)$ of $G$ is $\bot$-internal if and only if $N(v) \setminus \ens{u}\subseteq N(N(u) \setminus \ens{v})$.
\end{lemma}

\begin{proof}{
Let us consider a link $(u,v) \in E$ and let $G'=G - (u,v)$ be the bipartite graph obtained by removing the link $(u,v)$ from $G$.
Then, by definition, \mbox{$E_\bot = E'_\bot \cup \{(u,x),\ x \in N(v) \setminus \ens{u}\}$}.

Suppose that $(u,v)$ is a $\bot$-internal link, \ie\ $E_\bot = E'_\bot$.
Then all links $(u,x)$ in the expression above already belong to $E'_{\bot}$.
Therefore, for each $x \in N(v) \setminus \ens{u}$, $\exists \ y \not= v \in \top$ such that $y \in N(u) \cap N(x)$. 
By symmetry, $x \in N(y)$ and $y \in N(u) \setminus \ens{v}$ therefore, $x \in N(N(u) \setminus \ens{v})$ and so $N(v) \setminus \ens{u} \subseteq N(N(u) \setminus \ens{v})$.

Suppose now that $N(v)\setminus \ens{u} \subseteq N(N(u)\setminus \ens{v})$. 
Then for each node $x \in N(v) \setminus \ens{u}$, $ \exists \ y \in N(u)\setminus{} \ens{v} $
such that $x \in N(y)$.
Thus, by definition of the projection, $(u,x) \in E'_{\bot}$.
Therefore
$E_{\bot}=E'_{\bot}$ and the link $(u,v)$ is $\bot$-internal.}
\end{proof}

\section{Datasets \label{sec:data}}

Our aim is to evaluate the importance of internal pairs and links in large real-world graphs,
 rather than obtain specific conclusions in a particular context.
That is why we study various instances of real-world bipartite graphs, expecting to observe different behaviors. 
We present in this section the datasets we will use and summarize their general features (number of nodes and links).
The graphs under consideration are social ones connecting people ($\bot$-nodes) through events, groups or similar interests ($\top$-nodes).

\begin{itemize}
\item {\em Imdb-movies} \cite{barabasi1999emergence} is obtained from the {\em Internet Movie Database} (\url{www.imdb.com}): 
it features actors connected to the movies they played in.
$|\bot|=127,823$ actors, $|\top|=383,640$ movies, $|E|=1,470,418$.
\item {\em Delicious-tags} \cite{gorlitz2008} consists of {\em Delicious} (\url{www.delicious.com}) users connected to the tags they use for indexing their bookmarks.
$|\bot|=532,924$ users, $|\top|=2,474,234$ tags, $|E|=37,421,585$.
\item {\em Flickr-tags} \cite{prieur2008} consists of {\em Flickr} (\url{www.flickr.com}) users connected to the tags they use for indexing their photos.
$|\bot|=319,675$ users, $|\top|=1,607,879$ tags, $|E|=13,336,993$.
\item {\em Flickr-comments}: same as above, except that {\em Flickr} users are linked to the photos they comment.
$|\bot|=760,261$ users, $|\top|=12,678,244$ photos, $|E|=41,904,158$.
\item {\em Flickr-favorites}: same as above, except that users are linked to the photos they pick up as favorites.
$|\bot|=321,312$ users, $|\top|=6,450,934$ photos, $|E|=17,871,828$. 
\item {\em Flickr-groups}: same as above, except that users are linked to the groups they belong.
$|\bot|=72,875$ users, $|\top|=381,076$ groups, $|E|=5,662,295$.
\item {\em P2P-files} \cite{aidouni2009} is obtained from {\em peer-to-peer} file exchange \textit{eDonkey}:
users are linked to the files they provide. 
$|\bot|=122,599$ peers, $|\top|=1,920,353$ files, $|E|=4,502,704$.
\item {\em PRL-papers} has been extracted from the {\em Web of Science} database (\url{www.isiwebofknowledge.com}), collecting papers and authors of \textit{Physical Review Letters} from 2004 to 2007.
$|\bot|=15,413$ authors, $|\top|=41,633$ papers, $|E|=249,474$.
\end{itemize}


\section{Analysis of real-world cases \label{sec:meas}}

In this section, we use the notions of internal links and pairs introduced in Section~\ref{sec:def} to describe the real-world cases presented in Section~\ref{sec:data}. Let us insist on the fact that our aim is {\em not} to provide accurate information on these specific cases, but to illustrate how internal links and pairs may be used to analyze real-world data. We first show that there are many internal links in typical data, then study the number of internal links of each node and the correlation of this number with the node's degree.



Since the links attached to $\top$-nodes (resp. $\bot$-nodes)
of degree 1 are all $\bot$-internal (resp. $\top$-internal),
and since there may be a large fraction of nodes with degree 1 in real-world graphs,
we only study in the sequel links attached to nodes with degree at least 2.

\subsection{Amount of internal links and pairs}



In order to capture how redundant is the bipartite structure, we compute the number of $\top$- and $\bot$-internal pairs and links. 
The fraction of internal links, denoted $f_{E_I}$ and presented in Table~\ref{basic} seems in general not negligible.
A  quantitative analysis of these values however requires the definition of a benchmark.
That is why we compare the measures to the corresponding amounts on random bipartite graphs with the same sizes and degree distributions, which is a typical random model to evaluate the deviation from an expected behavior -- see for example \cite{newman2001random,newman2003social}.
The measures related to this model will be referred to with the symbol *.

We denote by $\mathcal{P}_I(\bot)$ (resp. $\mathcal{P}_I(\top)$) the set of $\bot$-internal pairs (resp. $\top$-internal pairs)  and by 
$E_I(\bot)$ (resp. $E_I(\top)$) the set of $\bot$-internal links (resp. $\top$-internal links).
We normalize the number of internal pairs and links measured on real graphs  to the values obtained with the model described above.
The corresponding results are also presented in Table~\ref{basic}.


\begin{table}[h!]
\center
\begin{tabular}{l| c c c | c c c}
& $ f_{E_I}(\bot)$
& $ \frac{\mathcal{P}_I(\bot)}{\mathcal{P}^*_I(\bot)} $ 
& $ \frac{E_I(\bot)}{E^*_I(\bot)} $
& $ f_{E_I}(\top)$
& $ \frac{\mathcal{P}_I(\top)}{\mathcal{P}^*_I(\top)} $ 
& $ \frac{E_I(\top)}{E^*_I(\top)} $\\
\hline
\multicolumn{1}{l|}{\em Imdb-movies} & 0.031 & 0.441 & 47.0 & 0.026 & 0.491 & 147 \\
\multicolumn{1}{l|}{\em Delicious-tags} & 0.112 & 0.972  & 1.47  & 0.104 & 1.823 & 5.31 \\
\multicolumn{1}{l|}{\em Flickr-tags} & 0.117 & 0.920 & 1.51 & 0.048 & 1.040 & 2.50 \\
\multicolumn{1}{l|}{\em Flickr-comments} & 0.398 & 0.258 & 4.22 & 0.002 & 0.151 & 22.0 \\
\multicolumn{1}{l|}{\em Flickr-groups}  & 0.228 & 0.491 & 2.21  & 0.015 & 0.249 & 2.86 \\ 
\multicolumn{1}{l|}{\em Flickr-favorites} & 0.172 & 0.574 & 2.02 & 0.002 & 0.704 & 12.4 \\
\multicolumn{1}{l|}{\em P2P-files} & 0.337 & 0.082 & 8.53 & 0.136 & 0.092 & 1430 \\
\multicolumn{1}{l|}{\em PRL-papers} & 0.718 & 0.033 & 7.17 & 0.487 & 0.001 & 11.2 \\
\end{tabular}
\caption{
Fraction of internal links ($f_{E_I}$), number of internal pairs ($\mathcal{P}_I$) and internal links ($E_I$) of real-world graphs normalized to the values on random bipartite graphs with the same size and same degree distributions.}
\label{basic}
\end{table}

We first notice that the behaviors in regards to the amount of internal links are very heterogeneous.
Still some general trends can be underlined: 
in the random case, $\bot$- and $\top$-internal links are underestimated.
So, the probability of having nodes sharing the same neighborhood is higher in real graphs than in random ones.
We may indeed expect, for instance, that people participating to the same paper have a higher probability to be coauthors of another one  than a random pair of authors.

Meanwhile the numbers of internal pairs are generally overestimated in random networks. 
To understand this effect, let us consider the extreme case where two $\bot$-nodes in a graph have either exactly the same neighborhood,  or no common neighbors.
Then all links are $\bot$-internal, and the graph does not contain any internal pair.
This example suggests that the number of internal pairs is probably anti-correlated to the number of internal links. 

In general, there is a correlation between the fact that the number of internal links is underestimated in random graphs
and the fact that the number of internal pairs is overestimated,
but this correlation does not hold in all cases.
Moreover, there is no direct link between these observations and the sizes or average degrees of the considered graphs.


Finally, we observe a specific behavior for the two graphs which correspond to tagging databases,
\ie{} {\em Delicious-tags} and {\em Flickr-tags}.
For these graphs we observe the lowest gaps between the real and random cases for $\bot$-internal links
and the amounts of $\bot$-internal pairs are very close in the real and random cases.
Conversely, they are the only graphs for which the amount of $\top$-internal pairs is underestimated in random graphs.



Since we can observe a wide range of behaviors both for $\top$- and $\bot$- internal links and pairs,
we will restrict our analysis in the following to $\bot$- internal links and pairs for the sake of brevity.
We will see that this allows enlightening observations.

\subsection{Distribution of internal links among nodes}


The notion of internal links partitions the links of each node into
two sets: the internal ones and the others.
We now study how the fraction of internal links is distributed among nodes. 
On Figure~\ref{fig:cum_dist}, we plot the complementary cumulative distribution of the fraction of internal links per node for the datasets under study.
We also plot the complementary cumulative distribution for random graphs.

\begin{figure}[h!]
\begin{tabular}{M{3.5cm}M{3.5cm}M{3.5cm}M{3.5cm}}
{\em Imdb-movies:} & {\em Delicious-tags:} & {\em Flickr-tags:} & {\em Flickr-comments:}
\tabularnewline
\includegraphics[angle=-90,scale=0.16]{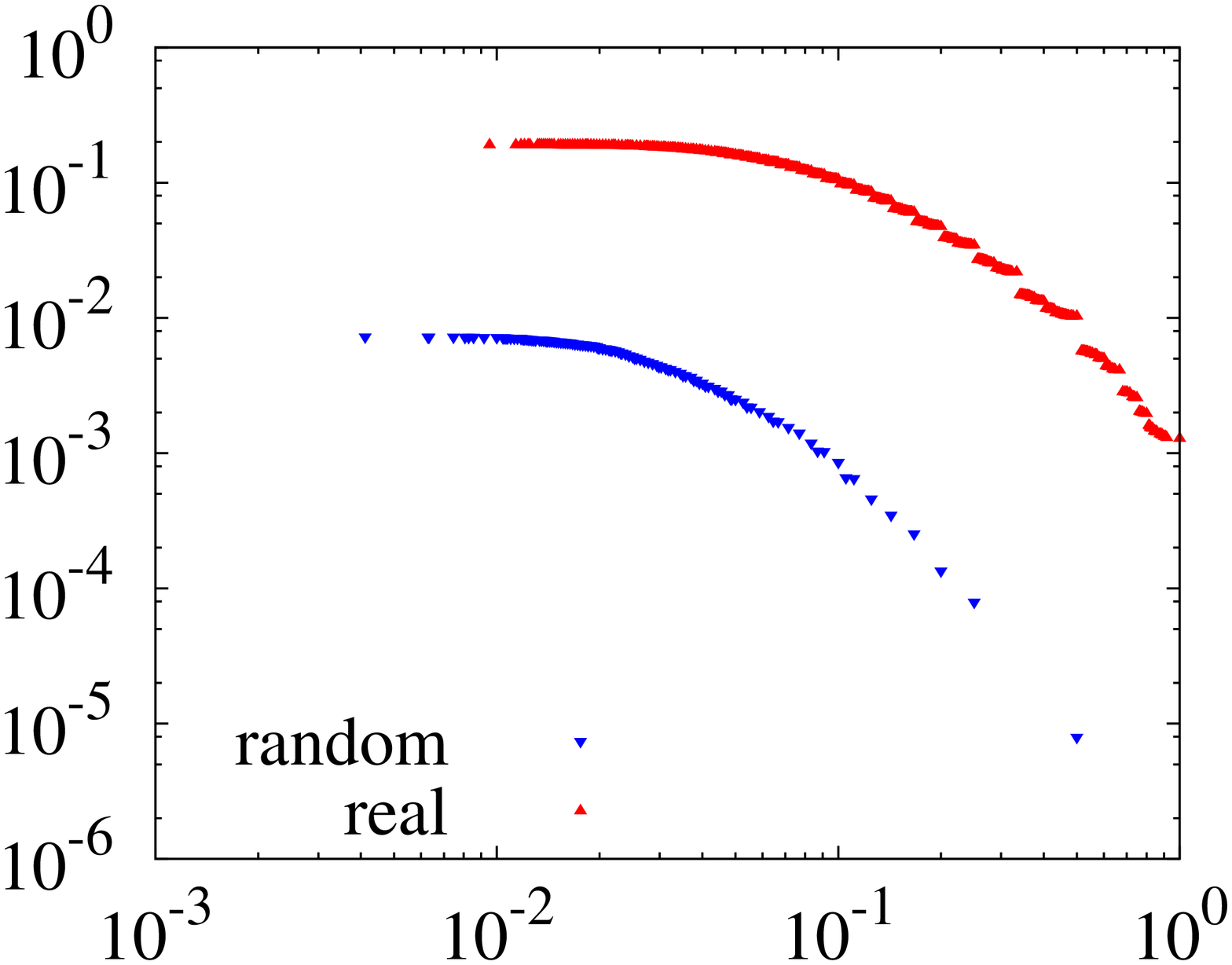} &
\includegraphics[angle=-90,scale=0.16]{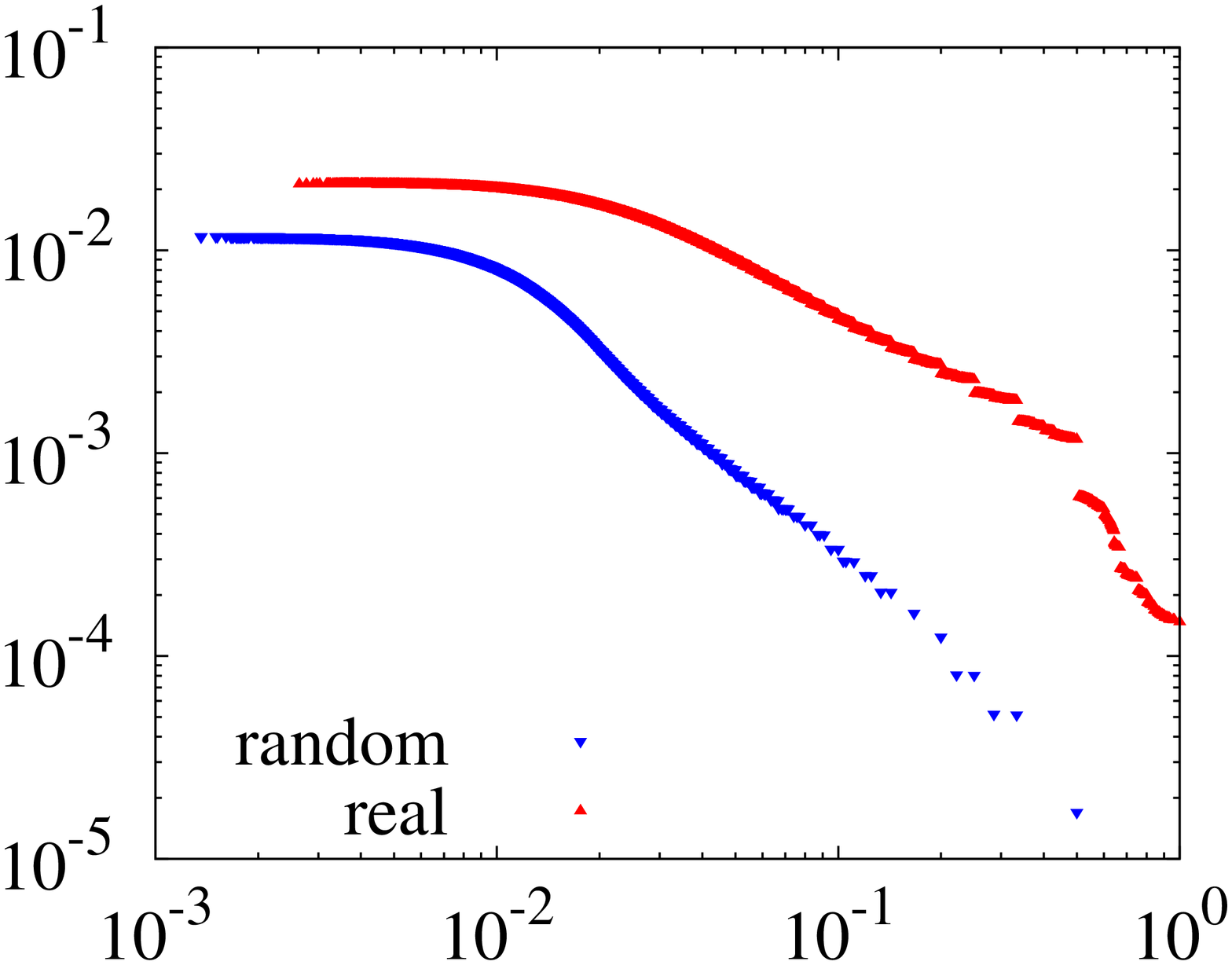} &
\includegraphics[angle=-90,scale=0.16]{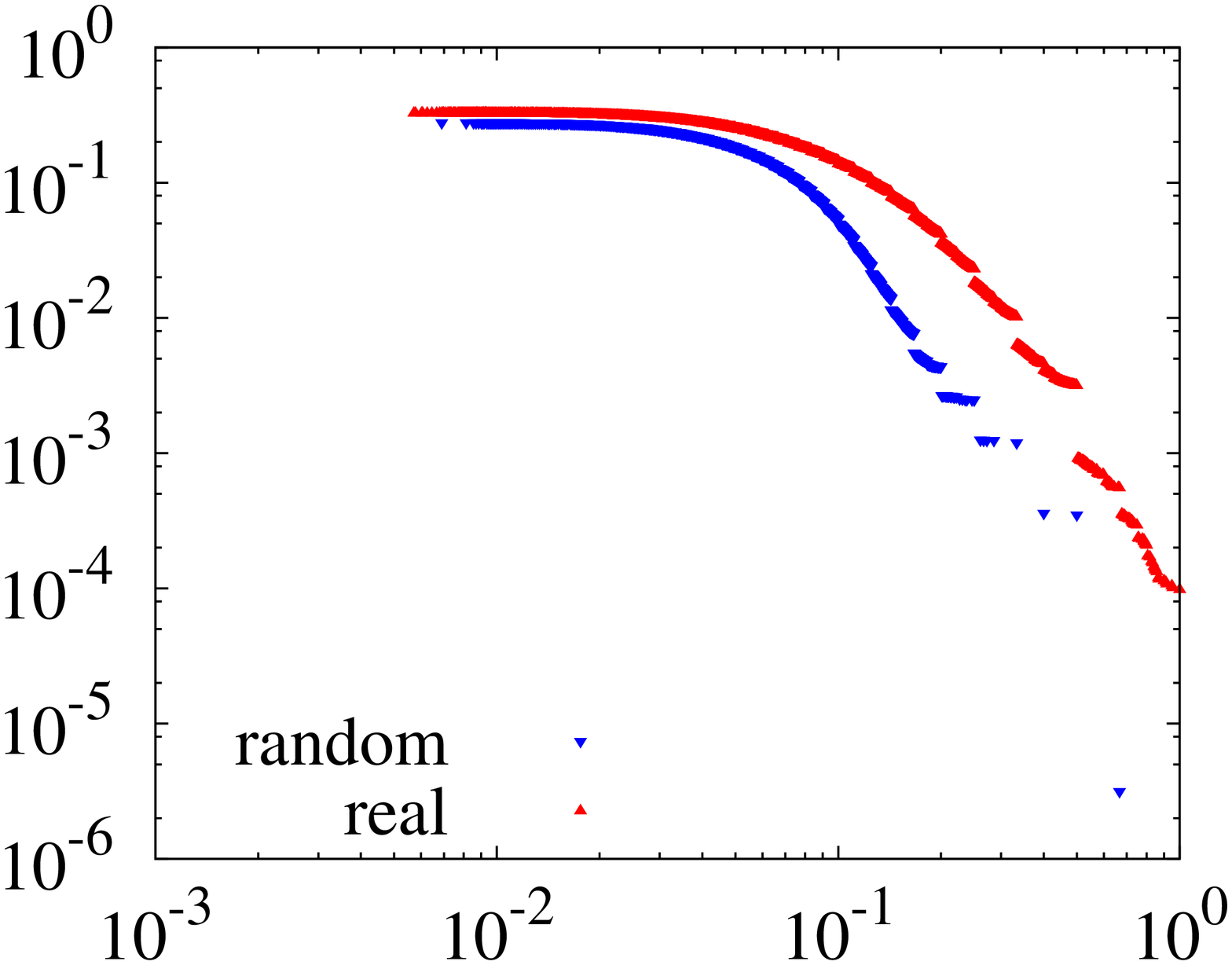} &
\includegraphics[angle=-90,scale=0.16]{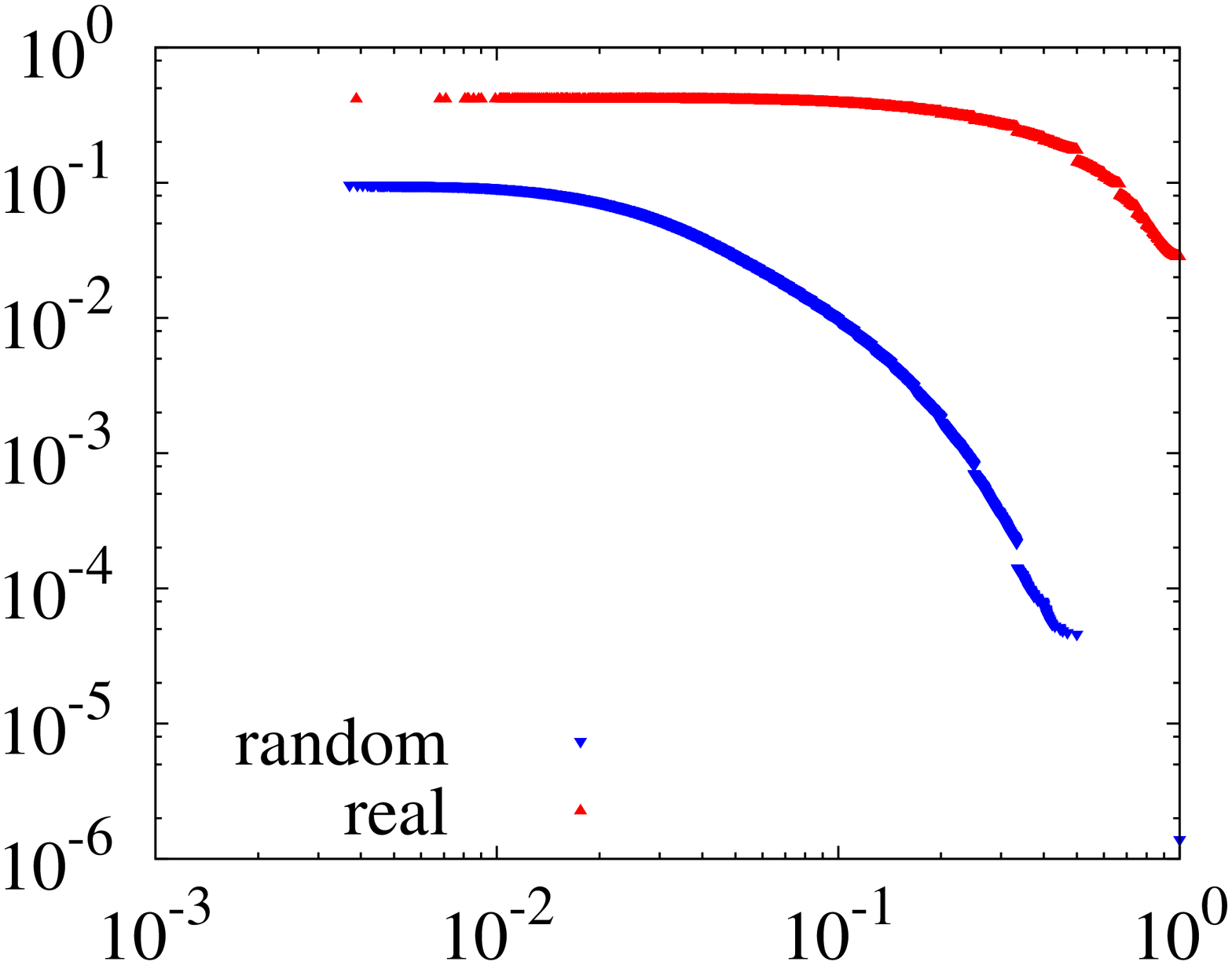}
\tabularnewline
\tabularnewline
{\em Flickr-groups:} & {\em Flickr-favorites:} & {\em P2P-files:} & {\em PRL-papers:}
\tabularnewline
\includegraphics[angle=-90,scale=0.16]{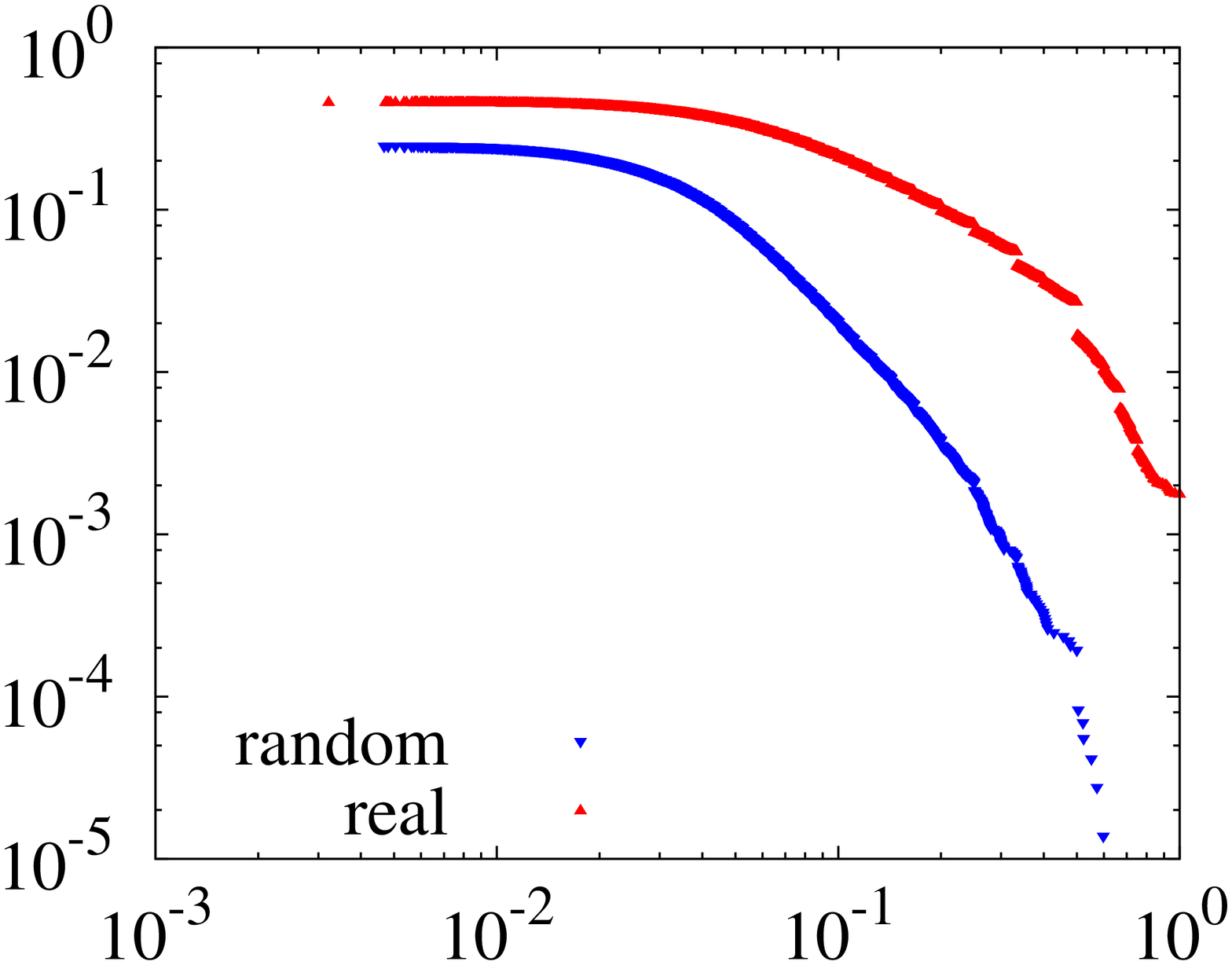} &
\includegraphics[angle=-90,scale=0.16]{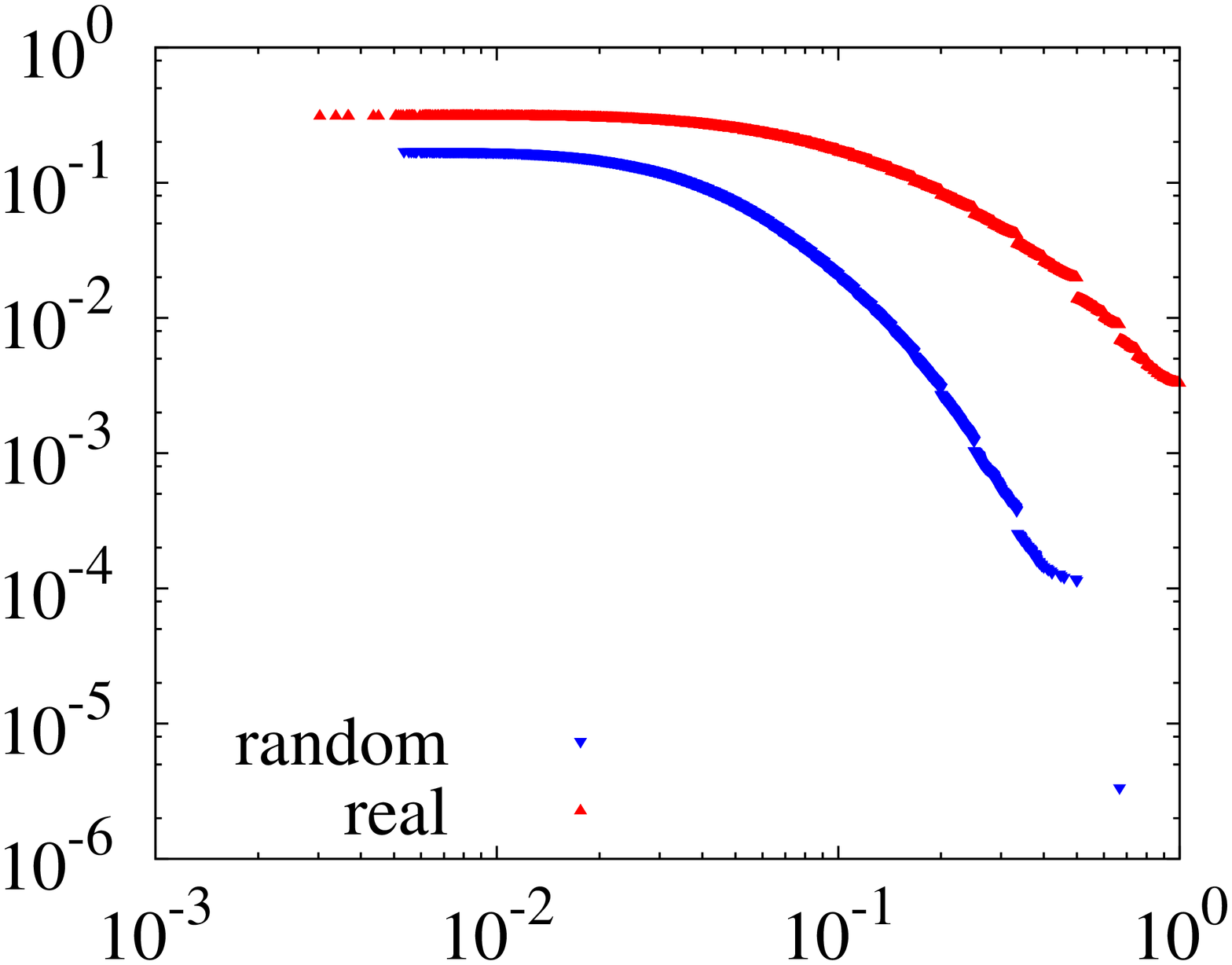} &
\includegraphics[angle=-90,scale=0.16]{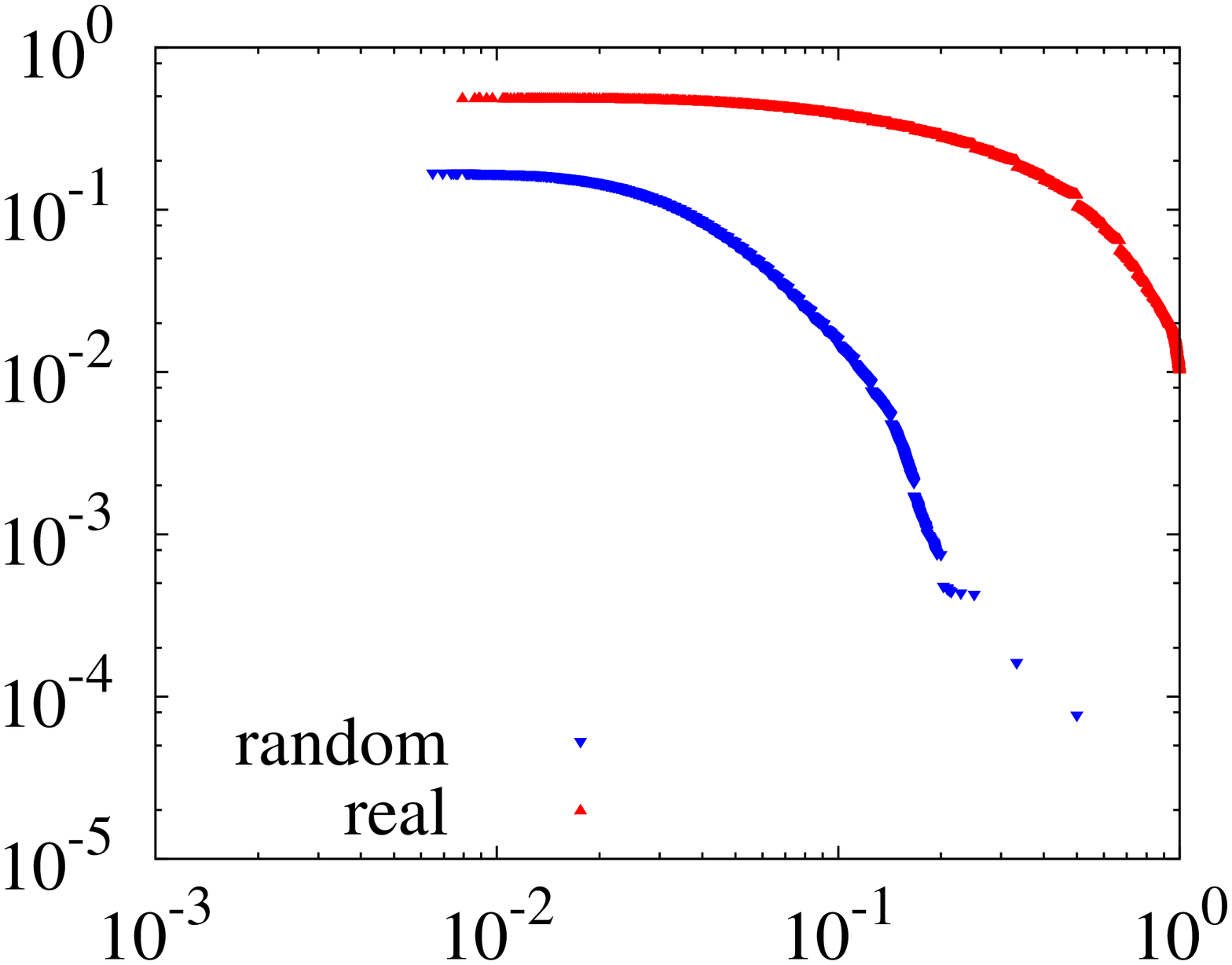} &
\includegraphics[angle=-90,scale=0.16]{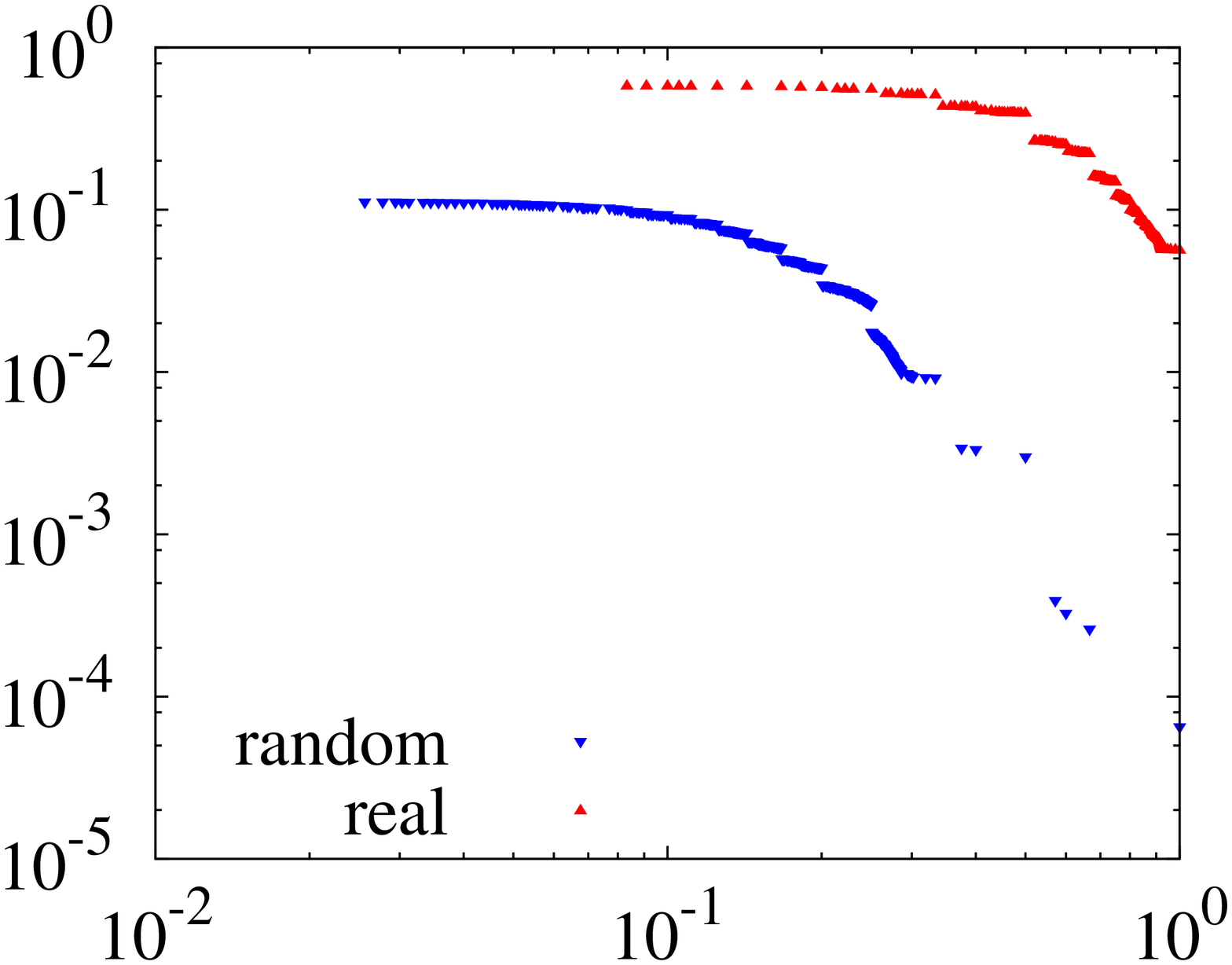}
\tabularnewline
\end{tabular}
\caption{Complementary cumulative distribution of the fraction of internal links per node.\label{fig:cum_dist}} 
\end{figure}

One of the most noticeable differences between both curves lies in the
probability of having a node whose links are all internal ($x=1$): this
fraction is indeed much higher in real than in random graphs.
  We also
observe that real graphs exhibit fewer nodes with very low (or null)
fractions of internal links
(though the fraction of nodes with {\em no} internal link is high in both cases).  In this respect too, the datasets behave
differently: for {\em Imdb-movies} the probability of having a $10^{-2}$ fraction of internal links is more than one order of
magnitude larger in the random than in the real graph, while
{\em  Flickr-tags} curves are close to be superimposed at low fractions.
Notice that this is not directly related to the fact that the number of internal links is
underestimated or not in random graphs:
for {\em Delicious-tags} the ratio between the number of $\bot$-internal links in the real and in the random case
is smaller than for {\em Flickr-tags},
but the difference between the distributions of the fraction of internal links per node
are larger for {\em Delicious-tags} than for {\em Flickr-tags}.

Finally, the very low fractions that we observe are associated to nodes with high degree:
to have a $ 10^{-4}$ fraction of internal links, a node has to have a degree of at least $ 10^{4}$.
Therefore, we study in the following the correlation between the degree of a node and its number of internal links.

\subsection{Correlation of internal links with node degrees} 


As stated before, the number of ($\bot$-)internal links of a node is called its ($\bot$-)internal degree, its total number of links being its degree. 
We investigate in this section the relationship between both quantities, plotting on
Figure~\ref{correlation} the average degree of a node in regards its
($\bot$-)internal degree for the real datasets and the randomized
ones.

\begin{figure}[h!]
\begin{tabular}{M{3.5cm}M{3.5cm}M{3.5cm}M{3.5cm}}
{\em Imdb-movies:} & {\em Delicious-tags:} & {\em Flickr-tags:} & {\em Flickr-comments:}
\tabularnewline
\includegraphics[angle=-90,scale=0.15]{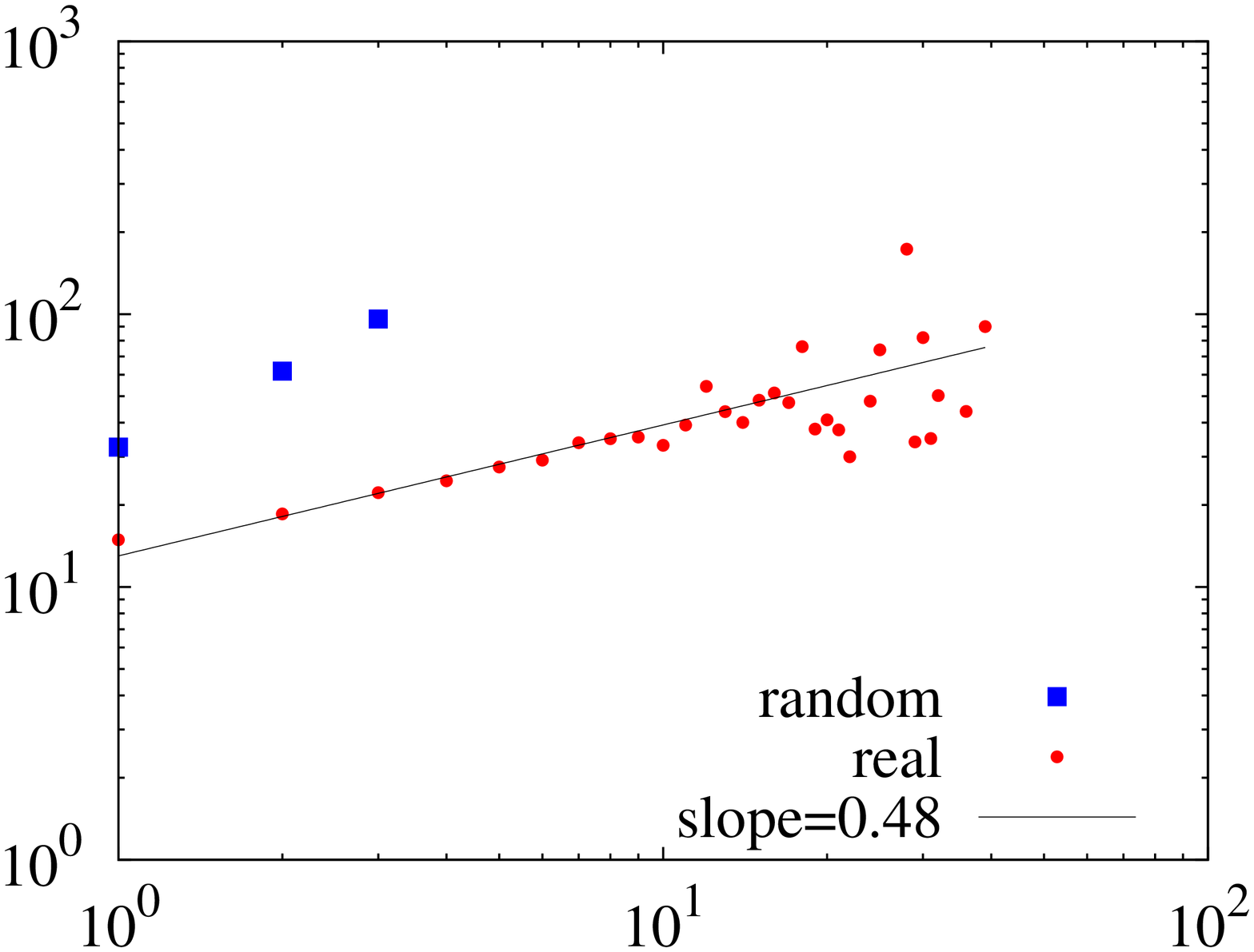} &
\includegraphics[angle=-90,scale=0.15]{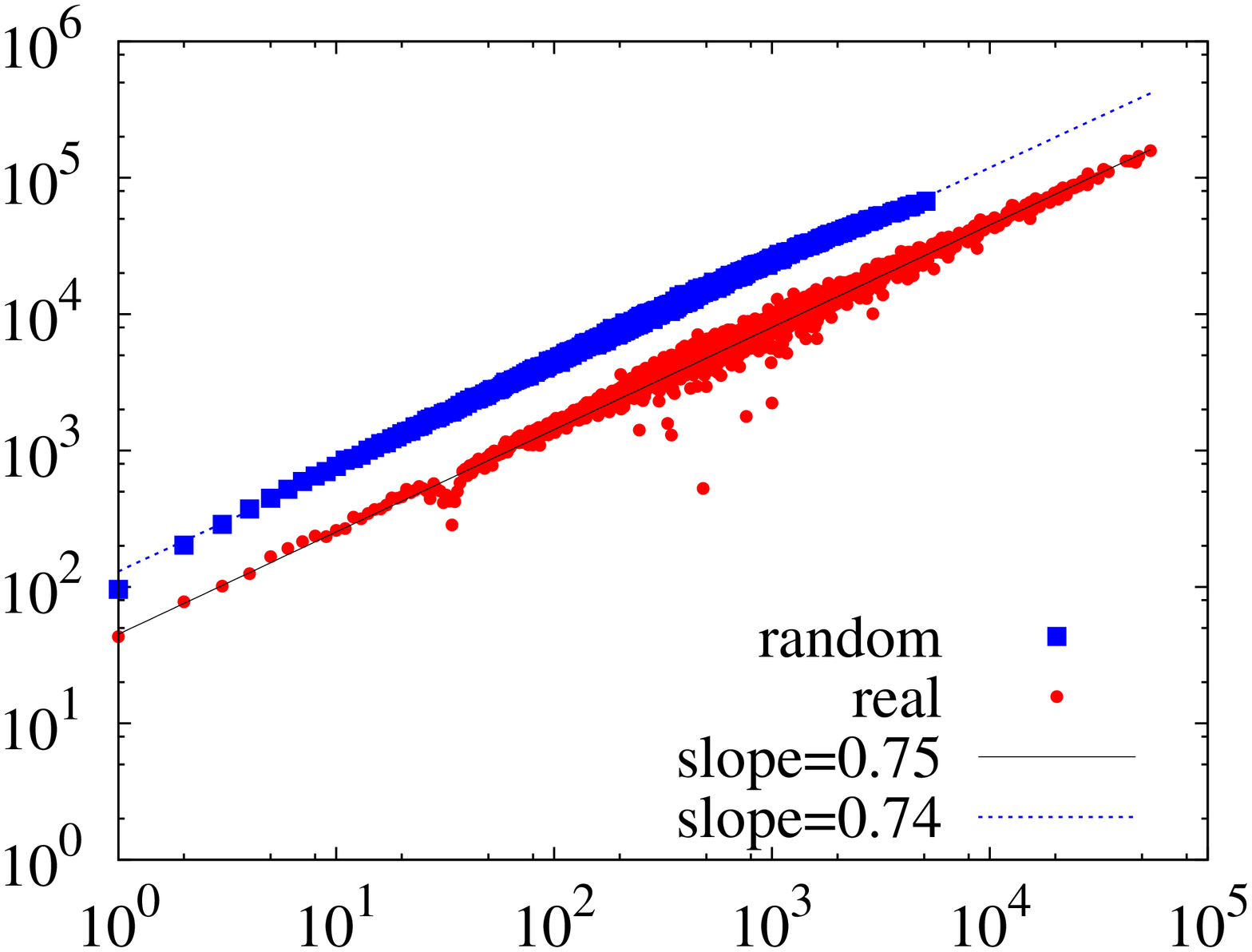} &
\includegraphics[angle=-90,scale=0.15]{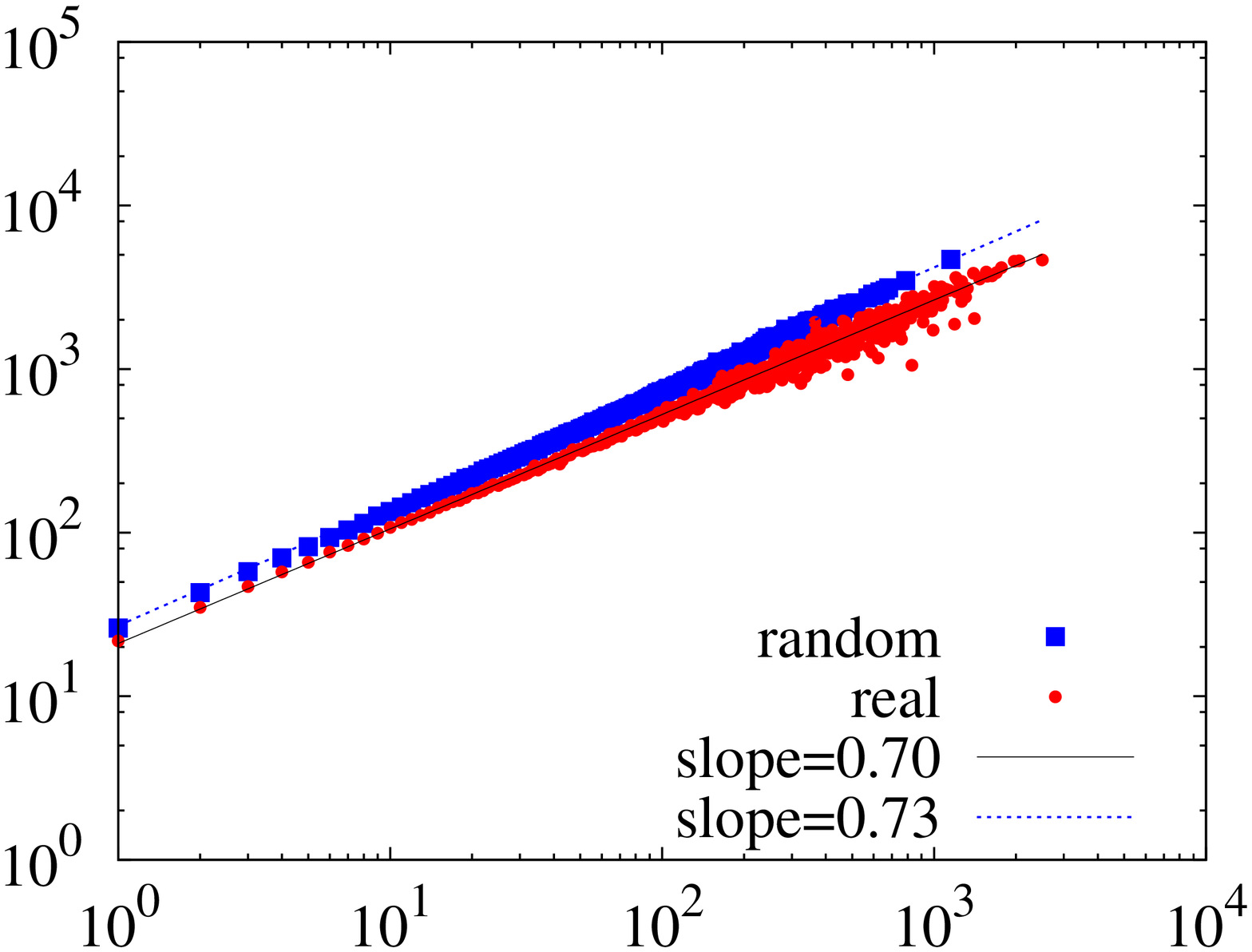} &
\includegraphics[angle=-90,scale=0.15]{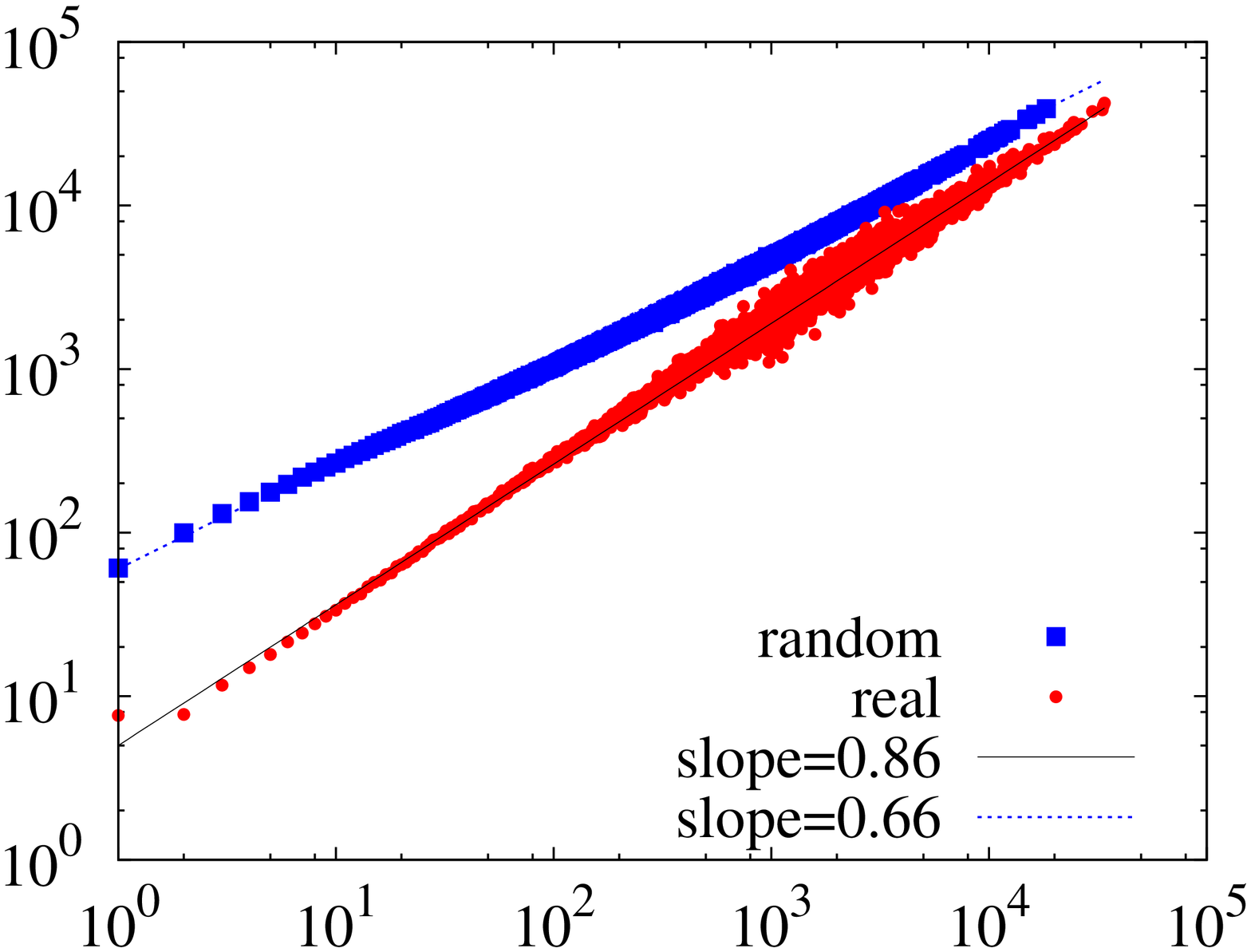}
\tabularnewline
\tabularnewline
{\em Flickr-groups:} & {\em Flickr-favorites:} & {\em P2P-files:} & {\em PRL-papers:}
\tabularnewline
\includegraphics[angle=-90,scale=0.15]{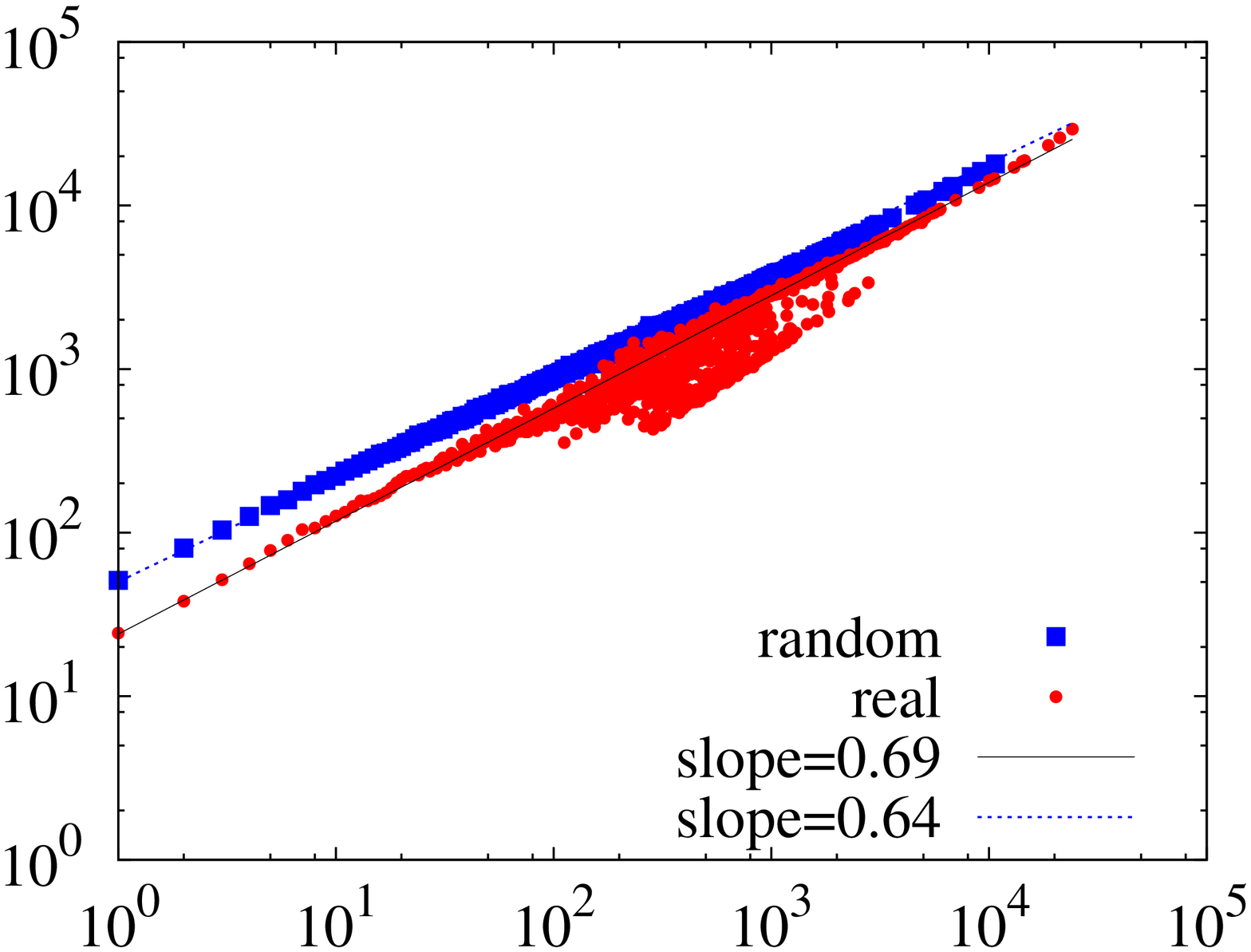} &
\includegraphics[angle=-90,scale=0.15]{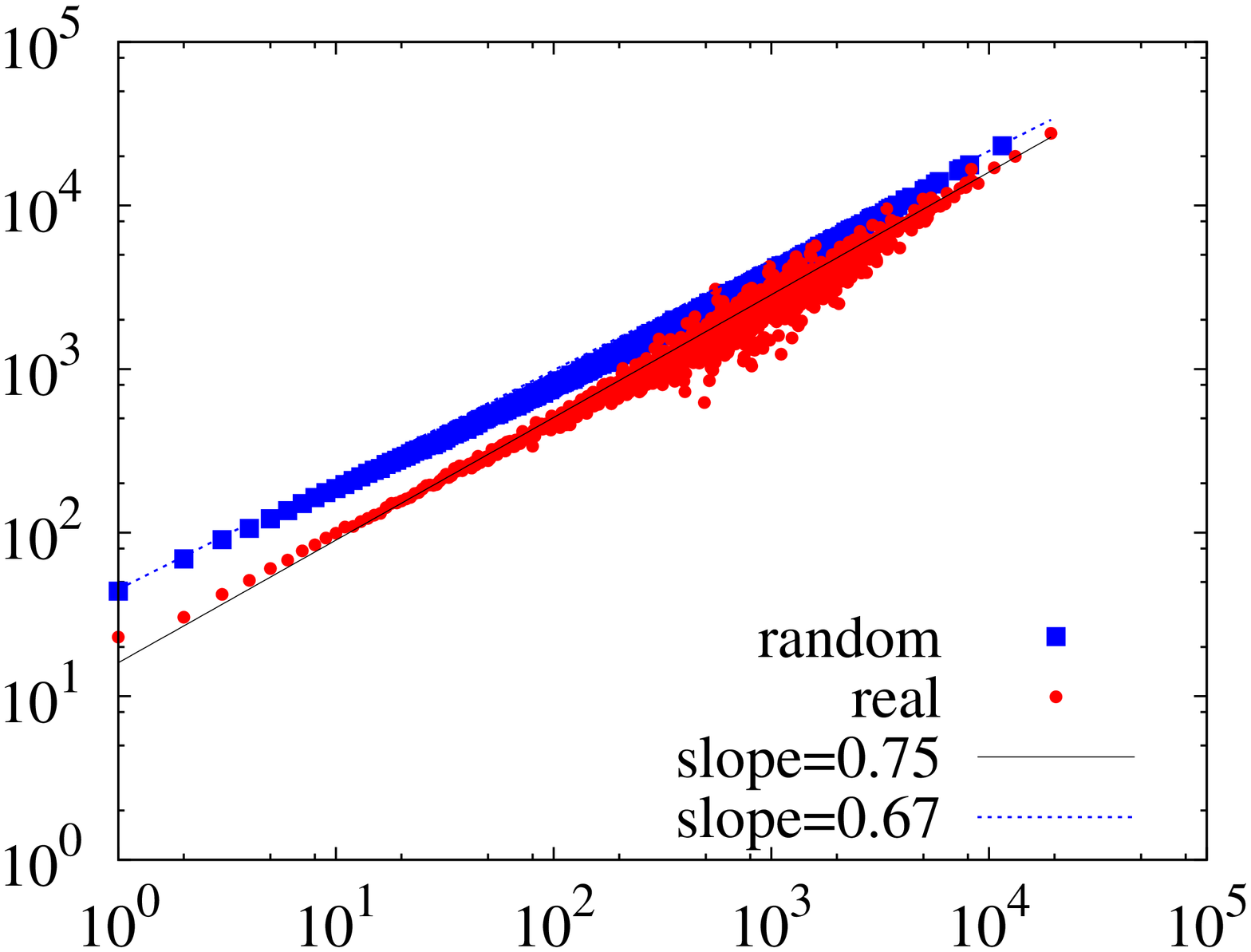} &
\includegraphics[angle=-90,scale=0.15]{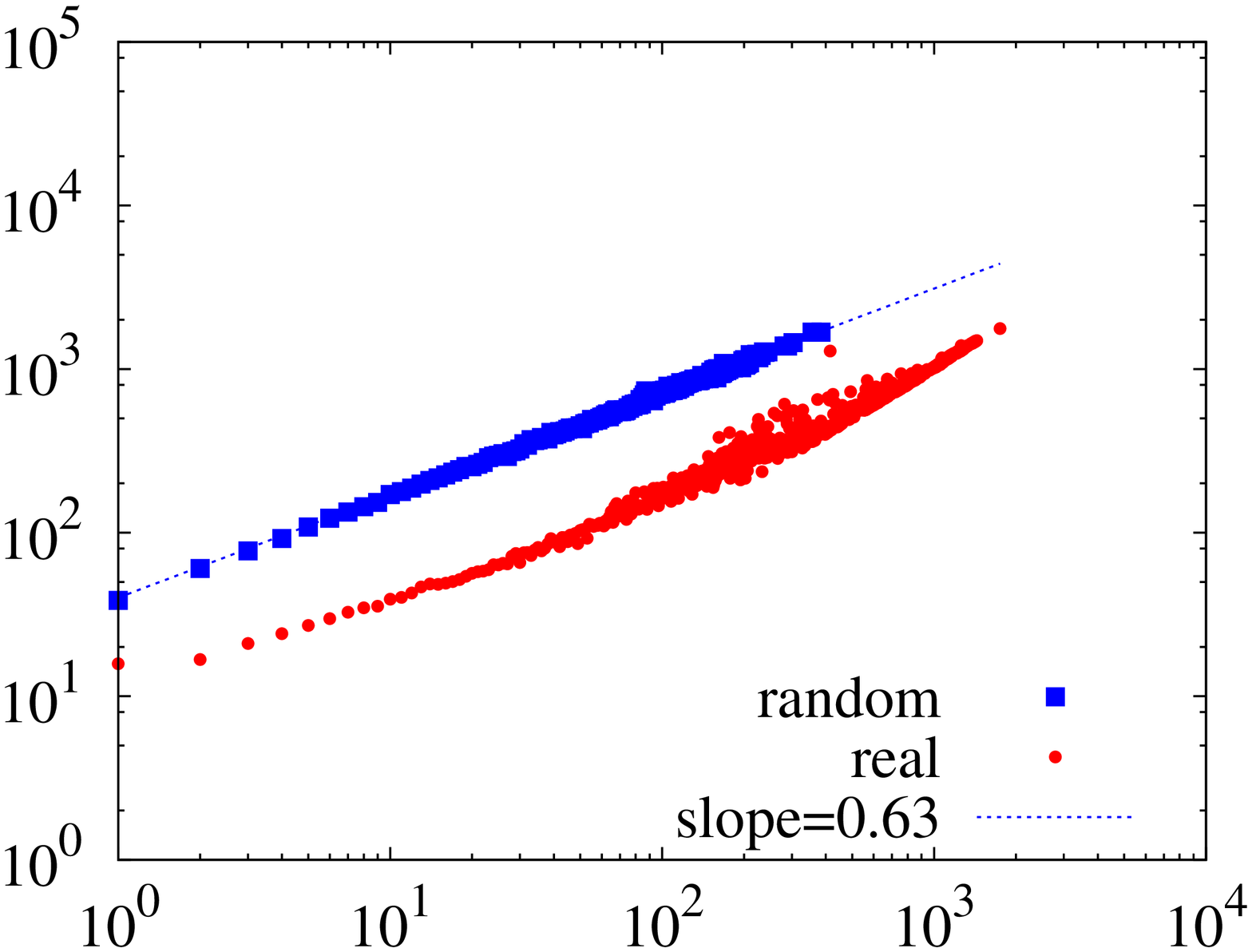} &
\includegraphics[angle=-90,scale=0.15]{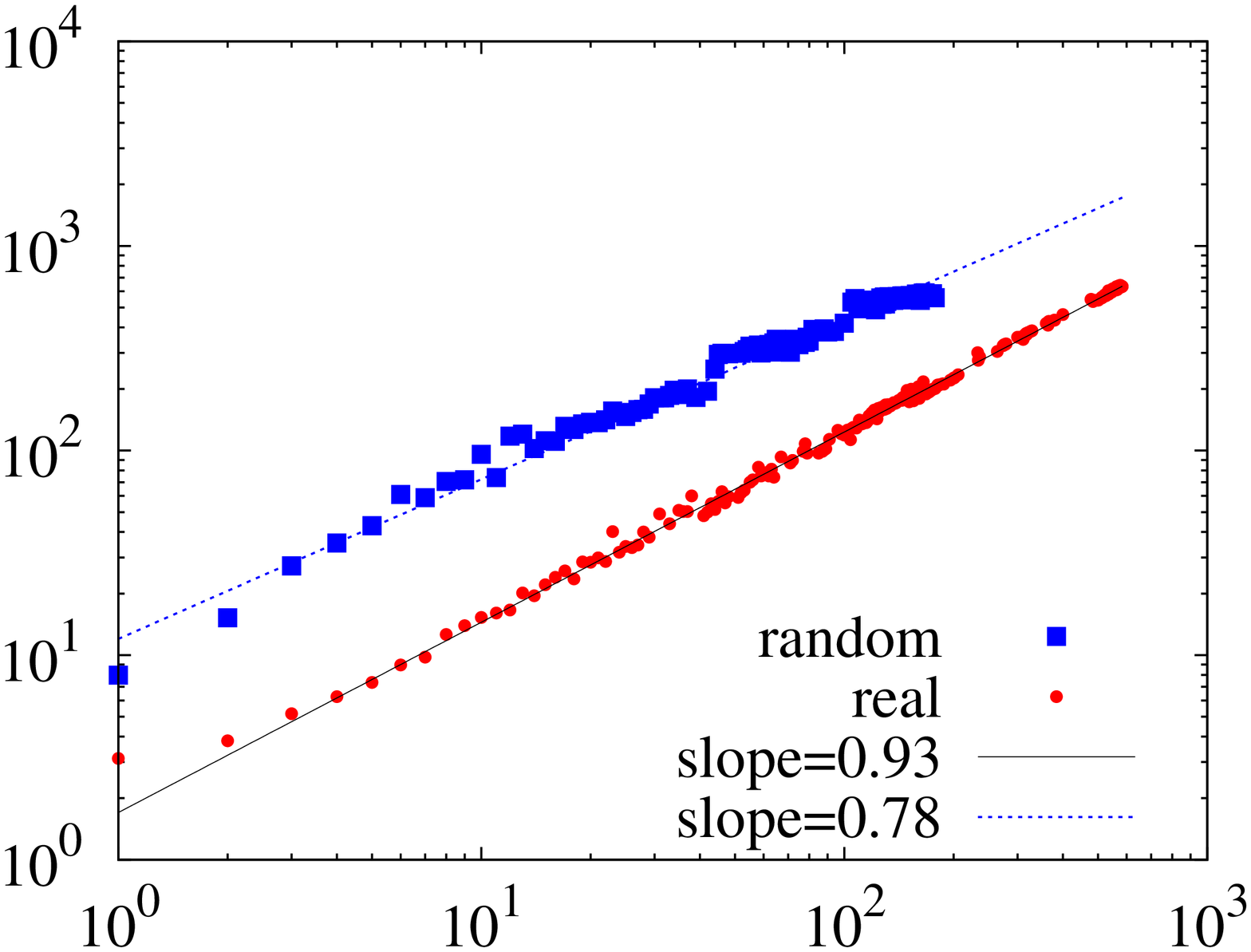}
\tabularnewline
\end{tabular}
\caption{Average degree as a function of the internal degree (for users projection).} 
\label{correlation}
\end{figure}

We observe that both real and random curves in several cases can be approximated by a sub-linear law on several decades.
However, this model is unsatisfactory on {\em P2P-files} database, and questionable on cases where the values are too rare or too scattered: most noticeably {\em Imdb-movies} and {\em Flickr-groups}.
The dispersion observed at large degrees is a consequence of the heterogeneous degree distribution, the number of nodes with high degree being low.
 
If the fact that a given link is internal or not was independent from the node's degree,
 these curves would be linear.
As random graphs have a sublinear behavior, that means that nodes with large degrees have on average a higher fraction of internal links.
This effect can be explained qualitatively: increasing the degree of a node $u$ - everything being otherwise unchanged - implies increasing the probability that one of his neighbors $v$ is such that
$N(v)\setminus \ens{u} \subseteq N(N(u) \setminus \ens{v})$.

On the other hand, the slope for real graphs is in most cases larger than for the random ones - again tagging datasets exhibit a different behavior. 
So there is an additional effect leading high degree nodes to have not as high an internal degree as expected by considering only the degree distributions.
This is consistent with previous observations: the real case provides more internal links and fewer nodes with a low (but not null) fraction of such links, which must be high degree nodes.
This stems from the fact that if nodes $u$ and $v$ are neighbors, the probability that $N(v) \setminus \ens{u} \subseteq N(N(u) \setminus{} \ens{v})$ is all the more important if $v$ has a small degree and $u$ a large one.
Therefore we expect that degree-correlated graphs yield larger slopes than degree-anticorrelated ones.
Yet, a more quantitative understanding of these phenomena calls for a study of the degree correlations in real-world graphs.




\section{Removing internal links\label{sec:delet}}

When modeling complex networks using bipartite graphs \cite{newman2001random,guillaume2004bipartite}, the presence of internal links may be a problem as they are poorly captured by models. To this regard, removing internal links before generating a random bipartite graph may lead to better models.
Moreover,  internal links are precisely these links in a bipartite graph which may be removed without changing the projection. As the bipartite graph may be seen as a compact encoding of its projection \cite{latapy2008basic}, one then obtains an even more compressed encoding. 
Considering the  example of the {\em P2P-files} dataset, it demands 30 MB if stored as a usual 2-mode table of lists, while the corresponding  $\bot$-projection (i.e. users) demands 213 MB and the $\top$-projection: 4.6 GB if stored as table of edges.

However, removing internal links is not trivial, as removing one specific link $(u,v$) may change the nature of other links: while they were internal in the initial graph, they may not be internal anymore after the removal of $(u,v)$. See Figure~\ref{fig:deletion} for an example. 
Therefore, in order to obtain a bipartite graph with no internal link but still the same projection (and so a {\em minimal} graph to this regard),
it is not possible in general to delete all initial internal links 
since this would alter dramatically the structure of the projection.
The set of internal links must therefore be updated after each removal. 
Going further, there may exist removal strategies which maximize the number of removals, whereas other may minimize it.

\begin{figure}[h!]
\centering
\begin{tabular}{M{4.5cm}M{4.5cm}M{4.5cm}M{4.5cm}}
\includegraphics[width=0.2\textwidth]{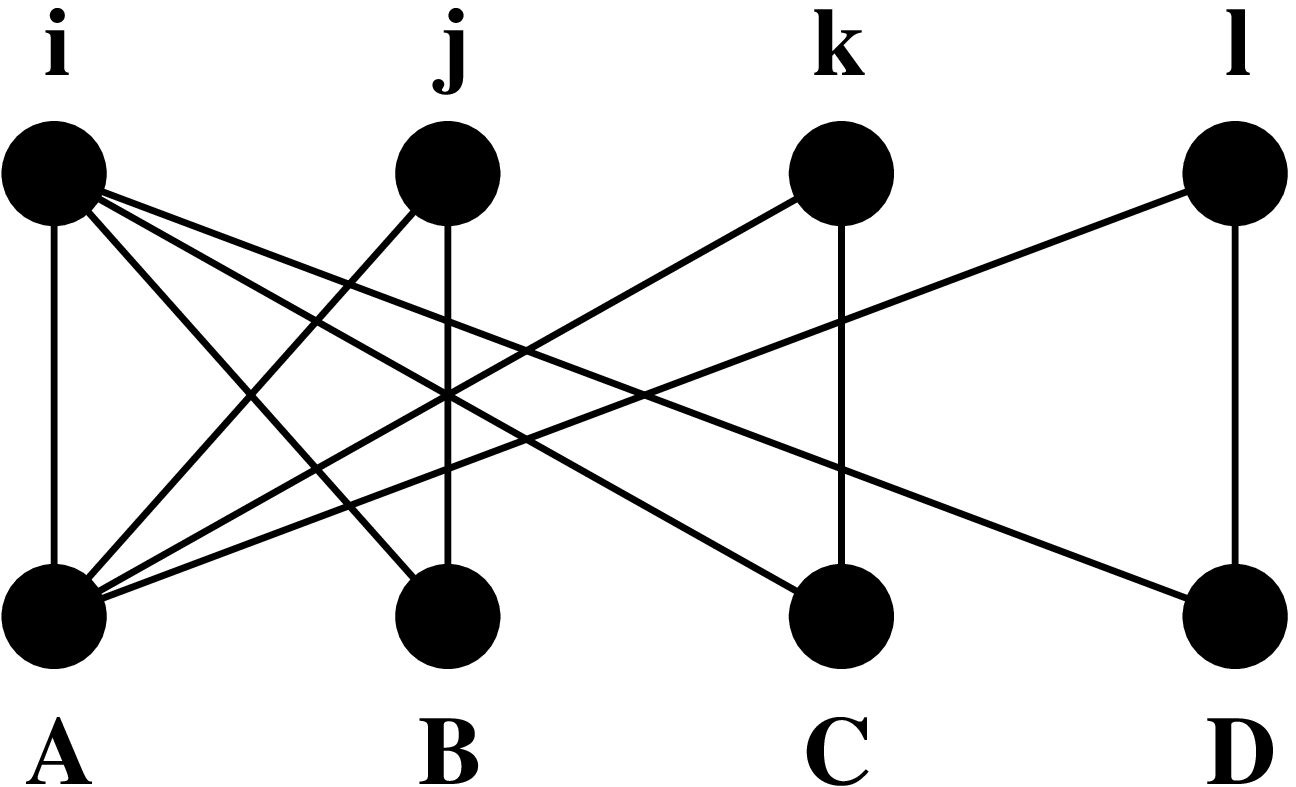} & 
\includegraphics[width=0.2\textwidth]{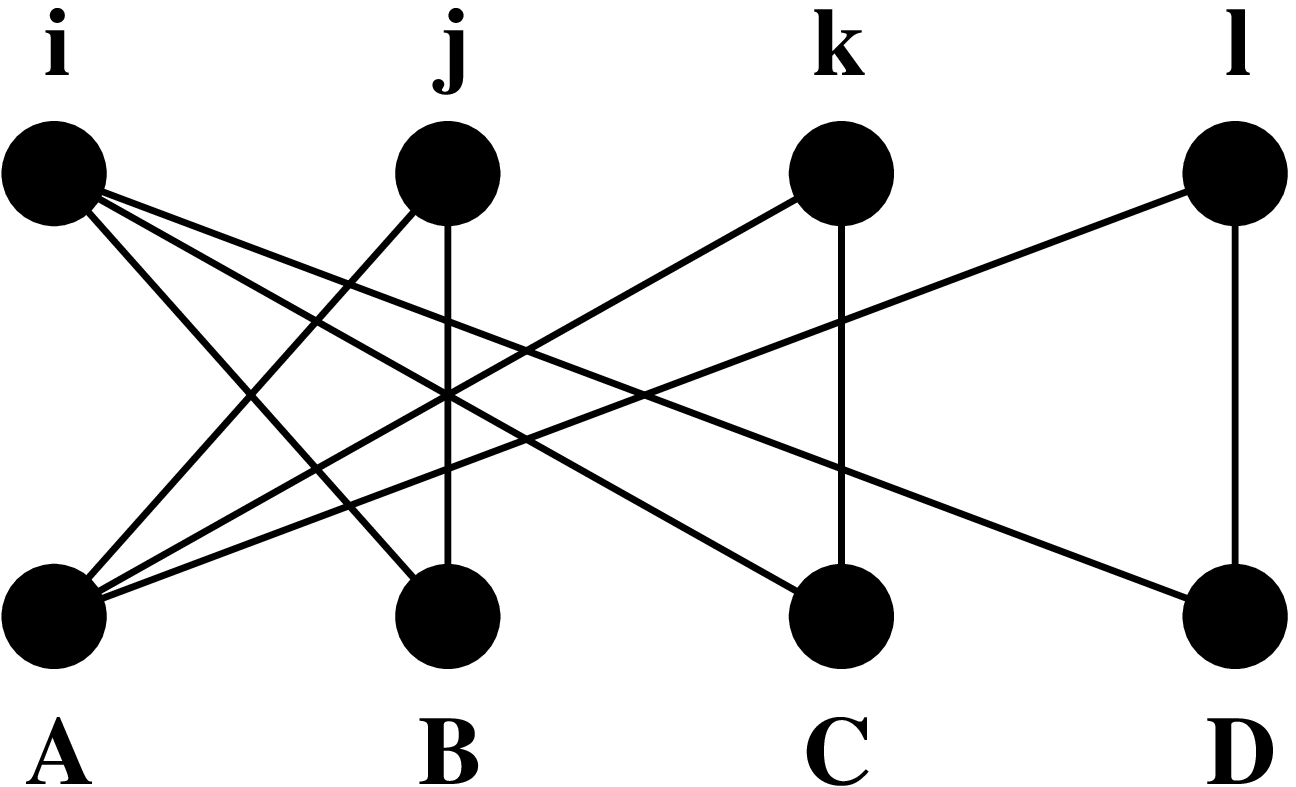} & 
\includegraphics[width=0.1\textwidth]{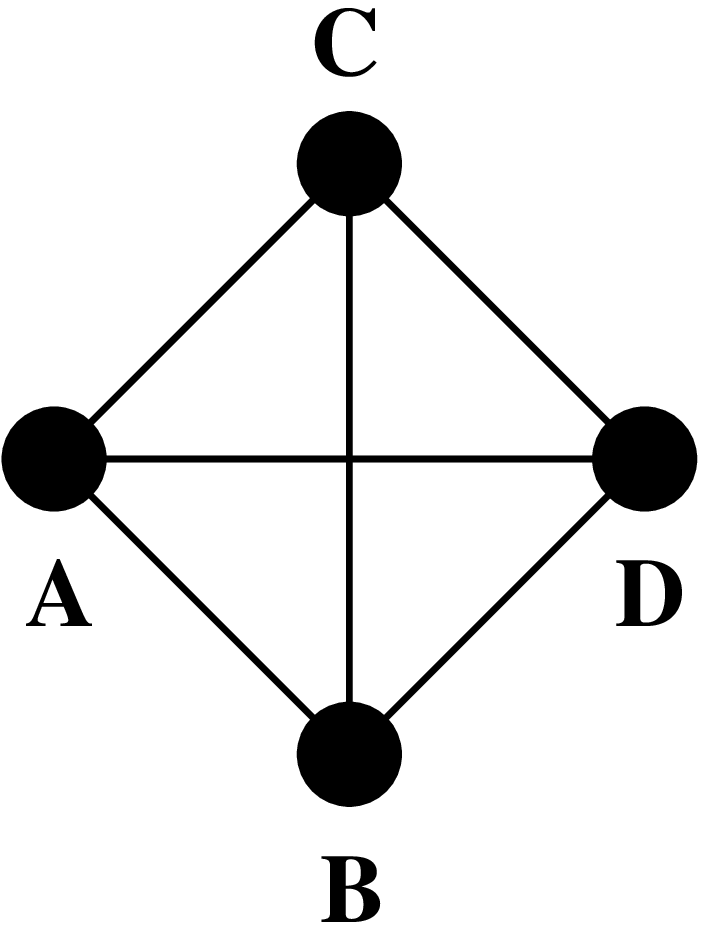}
\tabularnewline
$G$ & $G'=G-(A,i)$ & $G'_\bot = G_\bot$ 
\tabularnewline
\end{tabular}
\caption{\label{fig:deletion}\textbf{Influence of the deletion process on internal links.}
$\{ (A,i),(B,j),(C,k),(D,l) \} $ are $\bot$-internal links of $G$, 
yet deleting $(A,i)$ leads to $ G'$ where $\{ (B,j),(C,k),(D,l) \} $ are no longer $\bot$-internal links,
as they are the only links in $G'$ ensuring that $A$ is connected to respectively $B$, $C$ and $D$ in $ G_\bot$.} 
\end{figure}



To explore these questions, let us consider a random removal process,
where each step consists in choosing an  internal link at random and
removing it, and we iterate such steps
until no internal link remains. Figure~\ref{deletion-process} presents
the number of remaining internal links as a function of the number of
internal link removed for typical cases. We also plot the
upper bound $E_I - x$ (where $x$ denotes the number of link removals),
which represents the hypothetical case where all links initially internal remain
internal during the whole process.

\begin{figure}[h!]
\begin{tabular}{M{3.5cm}M{3.5cm}M{3.5cm}M{3.5cm}}
\multicolumn{2}{c}{\em Imdb-movies:} &  \multicolumn{2}{c}{\em P2P-files:} 
\tabularnewline
\tabularnewline
$\bot$-internal links: & $\top$-internal links: & $\bot$-internal links: & $\top$-internal links:
\tabularnewline
\includegraphics[angle=-90,width=0.26\textwidth]{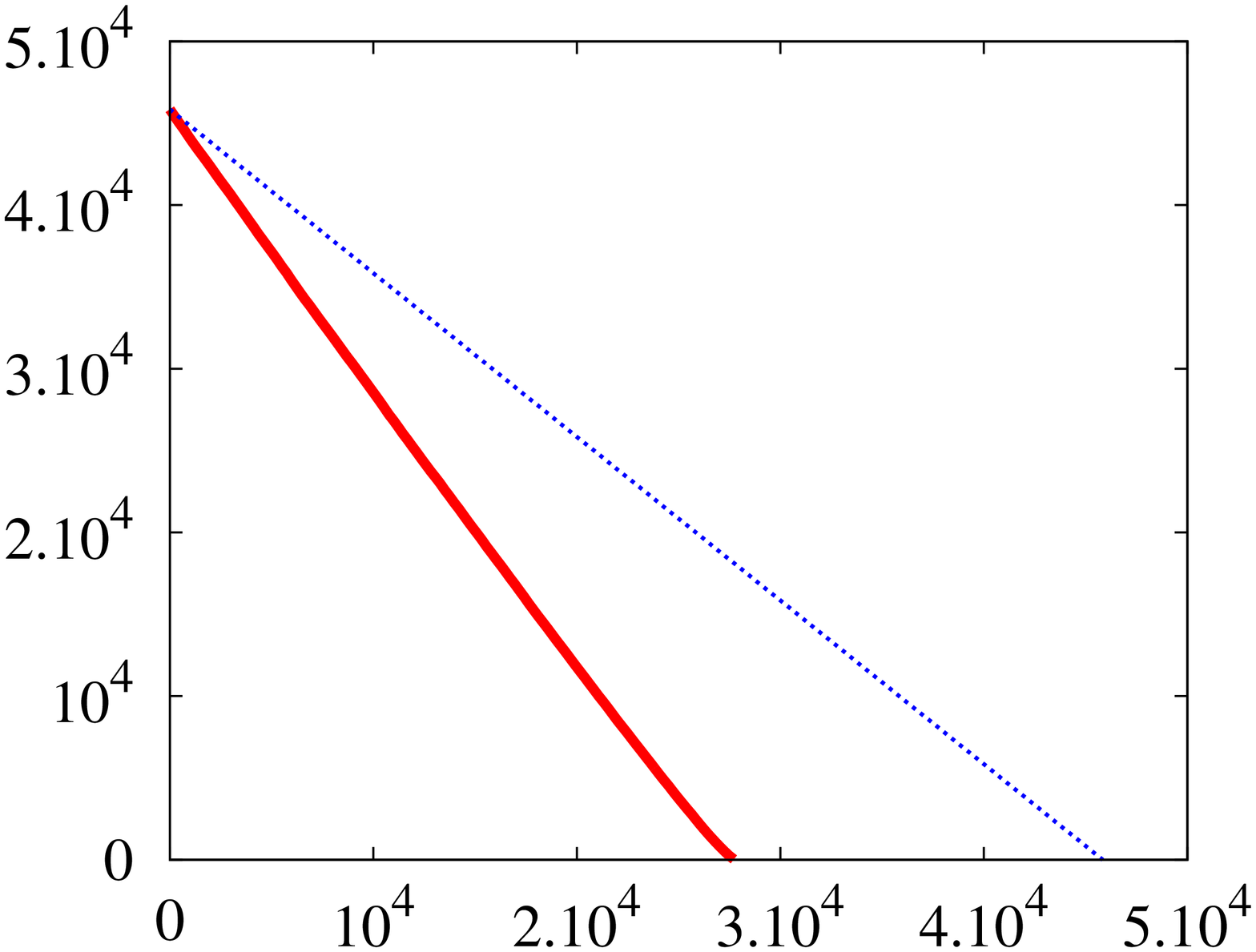}
& \includegraphics[angle=-90,width=0.26\textwidth]{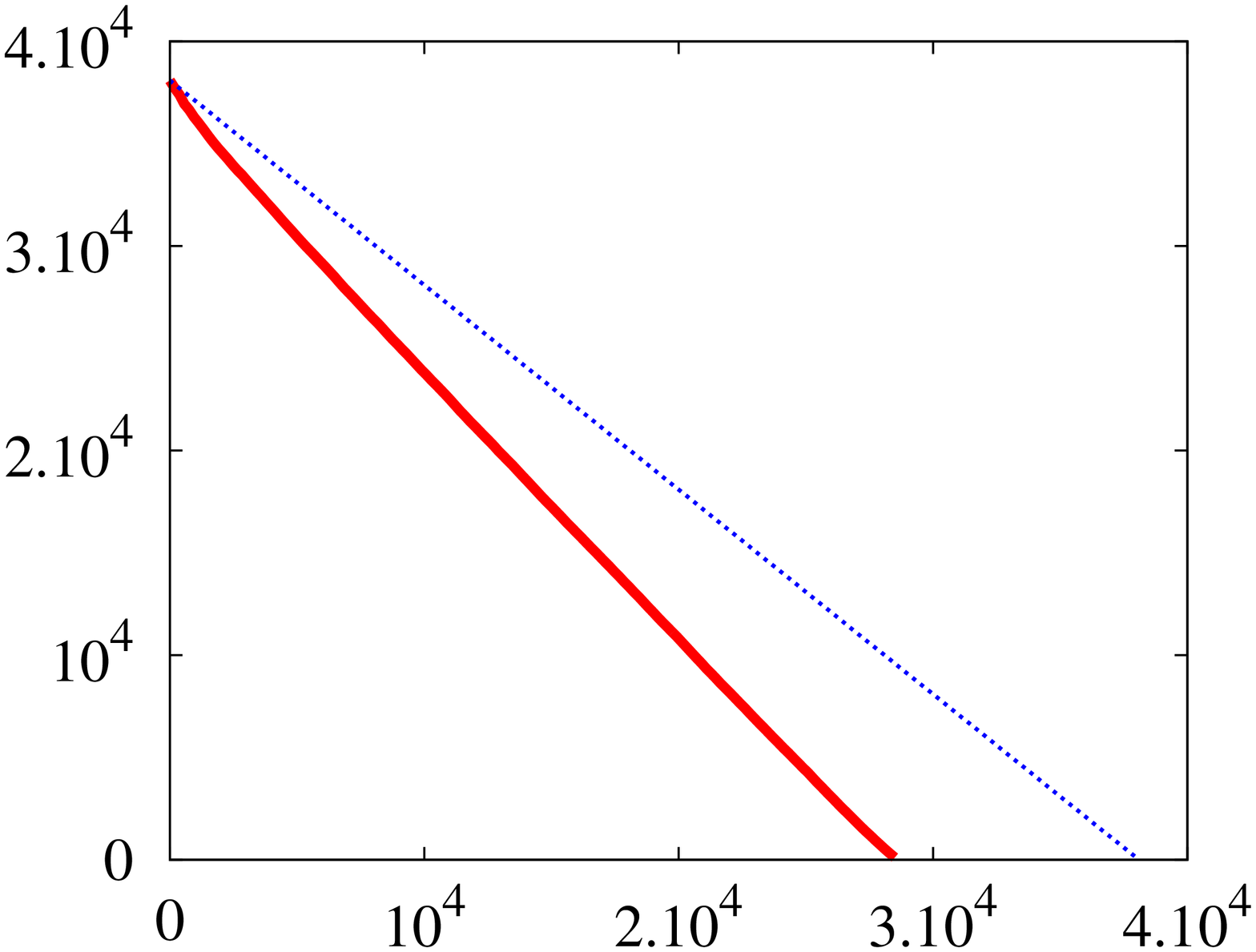}
& \includegraphics[angle=-90,width=0.26\textwidth]{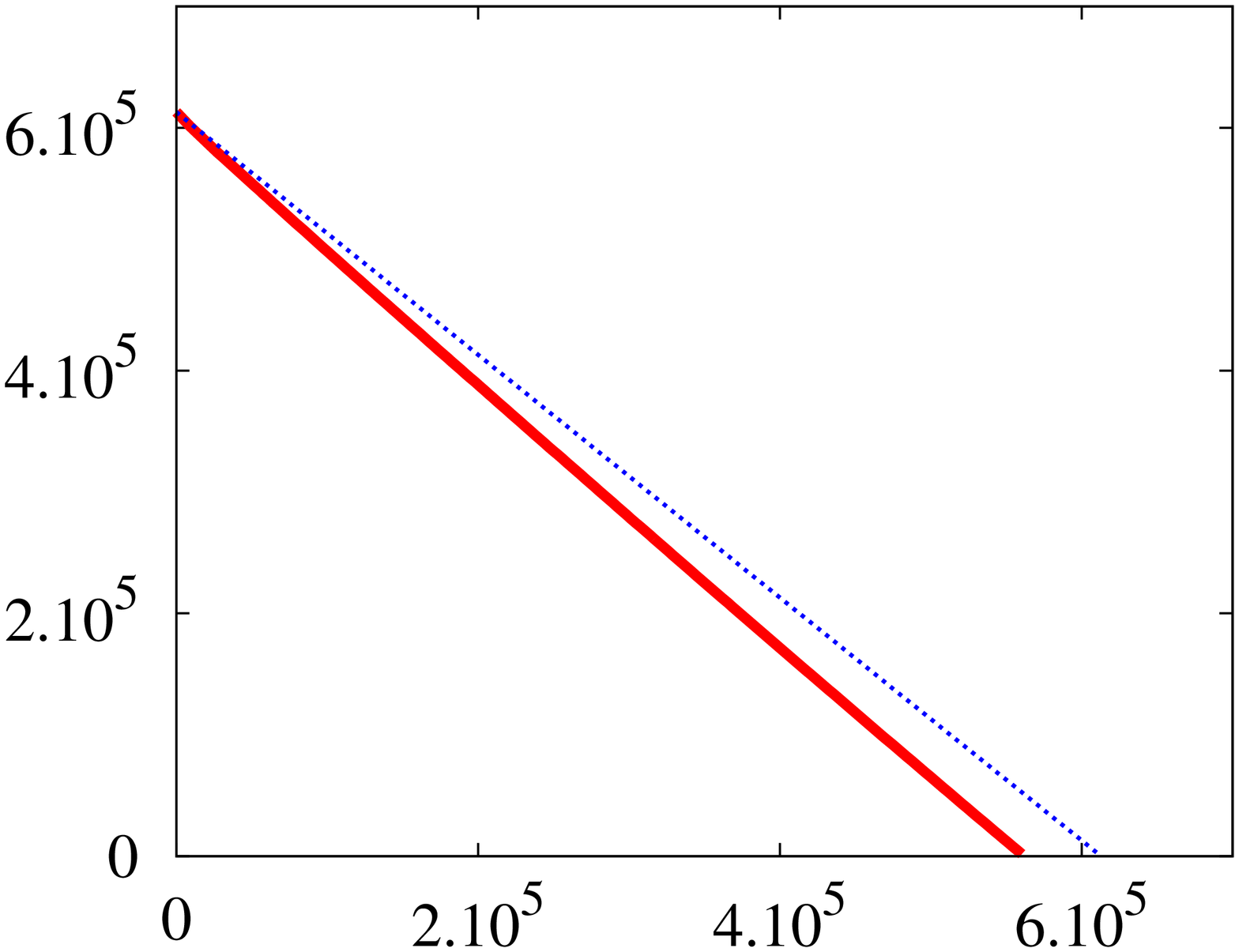} 
& \includegraphics[angle=-90,width=0.26\textwidth]{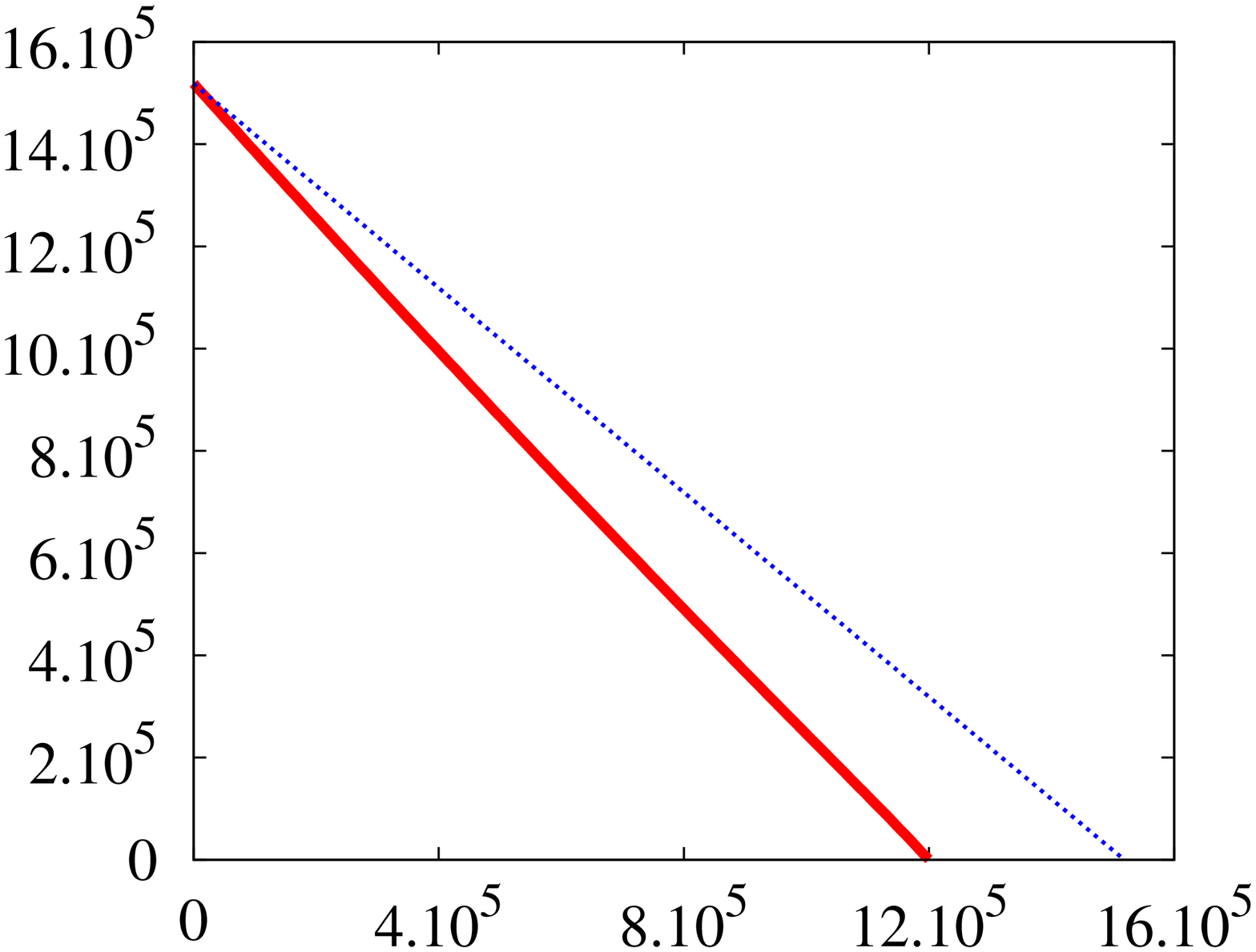}
\end{tabular}
\caption{Number of internal links remaining as a function of the number of deletions. Red thick line: random deletion process, blue thin line: theoretical upper bound.} 
\label{deletion-process}
\end{figure}

This random deletion process leads to a pruned bipartite graph, containing the information of the 1-mode graph.
Going back to the example of the {\em P2P} dataset, the obtained 2-mode storage graph demands 12 MB for the related $\bot$-projection and 22 MB for the $\top$ one,
thus enabling a compression to 0.40 (resp. 0.73) when compared to the standard 30 MB bipartite representation of the graph --- which is itself a compact encoding of the projections.

To go further, one may seek strategies that remove as many internal links as possible, for instance using a greedy algorithm selecting at each step the internal link leading to the lowest decrease of the number of remaining internal links. This is however out of the scope of this paper.

\section{Conclusion}

We introduced the notion of internal links and pairs in bipartite graphs, and proposed it as an important notion for analyzing real-world 2-mode complex networks. Using a wide set of real-world examples, we observed that internal links are very frequent in practice, and that associated statistics are fruitful measures to point out similarities and differences among real-world networks. This makes them a relevant tool for analysis of bipartite graphs, which is an important research topic. Moreover, removing internal links may be used to compact bipartite encodings of graphs and to improve their modeling.

We provided a first step towards the use and understanding of internal
links and pairs. Further investigations could bring us more precise
information about the role of internal links, in particular regarding
the dynamics. We suspect for instance that internal pairs may become
internal links with high probability in future evolution of the
graph. One may also study these links (and pairs) which are both
$\bot$- and $\top$-internal, as they may have  a special importance in a graph.

\subsection*{Acknowledgements}
This work is supported in part by the French ANR {\em MAPE} project.

\end{document}